\documentclass[11pt]{article}
\usepackage{style}
\usepackage[left=1in,bottom=1in,top=1in,right=1in]{geometry}

\hypersetup{%
	colorlinks   = true, %
	urlcolor     = blue, %
	linkcolor    = blue, %
	citecolor    = blue,  %
}

\title{Optimal Unambiguous DNFs and Alon-Saks-Seymour}
\author{Chirag Pabbaraju\thanks{Stanford University. Email: \texttt{cpabbara@cs.stanford.edu.}}}
\date{\today}

\begin{document}

\maketitle

\begin{abstract}
    We construct unambiguous DNFs having width $O(n)$ but $0$-certificate complexity $\Omega(n^2)$. By utilizing the special structure of these DNFs, we prove a lifting theorem with a constant-sized gadget that lifts the DNF to a communication problem, while losslessly translating the separation in certificate complexity to a separation in communication complexity. This leads to an optimal refutation of the Alon-Saks-Seymour conjecture, as well as an optimal communication lower bound for the Clique versus Independent Set problem, improving the previous results of Balodis, Ben-David, G\"{o}\"{o}s, Jain and Kothari (FOCS 2021, SICOMP 2023) by several doubly logarithmic factors. As further applications of our construction to query complexity and learning theory, we exhibit: (a) a family of Boolean functions that has an optimal quartic separation between certificate complexity and approximate degree, and (b) a sample compression lower bound of $\Omega(\sqrt{\log c})$ for multiclass concept classes over $c$ labels. 
\end{abstract}

\newpage

\section{Introduction}
\label{sec:intro}

We begin by describing a central puzzle studied in the work of \citet*{balodis2023unambiguous} to set up the context.

\paragraph{Unambiguous DNFs.} A DNF (Disjunctive Normal Form) $f=C_1 \vee \dots \vee C_m$ is a Boolean function that is the disjunction of terms $C_i$, where each $C_i=x_{i_1} \wedge \dots \wedge \neg x_{i_\ell}$ is a conjunction of literals (i.e., a variable or its negation). The DNF has term-width $\ell$ if every $C_i$ is a conjunction on at most $\ell$ literals. The DNF is \textit{unambiguous} if whenever $f=1$, exactly one term evaluates to $1$.

For any Boolean function on $n$ variables, a partial assignment $\rho \in \{0,1,*\}^n$ is a $0$-certificate for $x \in f^{-1}(0)$ if $\rho$ is consistent with $x$ (i.e., $\rho_i = x_i$ wherever $\rho_i \neq *$), and for any $x' \in \{0,1\}^n$ consistent with $\rho$, it holds that $f(x')=0$. We define the domain of $\rho$ to be $\dom(\rho)=\{i: \rho_i \neq *\}$, and the size of the certificate $\rho$ to be $|\rho|=|\dom(\rho)|$. The $0$-certificate complexity of $f$ at $x$, denoted $\C_0(f, x)$, is the size of the smallest 0-certificate for $x$. The $0$-certificate complexity of $f$, denoted $\C_0(f)$, is defined to be $\max_{x \in \{0,1\}^n}\C_0(f, x)$. The definitions for $1$-certificates and $1$-certificate complexity $\C_1(f)$ are analogous. Finally, the certificate complexity of $f$, denoted $\C(f)$, is defined as $\C(f)=\max(\C_0(f),\C_1(f))$.

The unambiguous $1$-certificate complexity of a function $f$, denoted $\UC_1(f)$, is the smallest $\ell$ such that $f$ is equivalent to an unambiguous DNF having term-width $\ell$. It was shown by \citet{goos2015lower} that for any function $f$, it holds that $\C_0(f) \le \UC_1(f)^2$. The first puzzle in \citet{balodis2023unambiguous} asks how large the separation can be in the opposite direction:

\paragraph{Puzzle I.} For $\alpha > 1$, does there exist a family of Boolean functions $f$ with $\C_0(f) \ge \Omega(\UC_1(f)^\alpha)$?
\\

Motivations for studying this puzzle arise from its implications to the Alon-Saks-Seymour conjecture \citep{kahn1991recent} and the classical Clique versus Independent Set problem \citep{yannakakis1988expressing} in communication complexity (see \Cref{sec:alon-saks-seymour-clique-independent-set} below). In light of the upper bound on $\C_0(f)$ above, the best possible exponent $\alpha$ for Puzzle I can be $\alpha=2$. Indeed, the exponent was confirmed to be 2 --- up to polylogarithmic factors --- by \citet{balodis2023unambiguous}, who constructed a family of functions $f$ satisfying $\C_0(f) \ge \widetilde{\Omega}(\UC_1(f)^2)$. More precisely, the separation shown in their work was (see Theorem 3.2 in \citet*{cheung2023online}):
\begin{align}
    \label{eqn:unambiguous-dnfs-previous-best-separation}
    \C_0(f) \ge {\Omega}\left(\frac{\UC_1(f)^2}{\log^6(n)}\right).
\end{align}
To our knowledge, this is the best known bound for Puzzle I prior to this work.

\subsection{Optimal Solution to Puzzle I}
\label{sec:puzzles-optimality}

Our first contribution is an optimal solution to Puzzle I which is completely devoid of any polylogarithmic factors.

\begin{restatable}[Unambiguous DNFs]{theorem}{UnambiguousDNF}
    \label{thm:unambigous-dnf}
    For any $n \ge 2$, there exists an unambiguous DNF $f:\{0,1\}^{n^2} \to \{0,1\}$ satisfying 
    \begin{align*}
        \UC_1(f)=O(n), \qquad \C_0(f)=\Omega(n^2).
    \end{align*}
\end{restatable}
By reparameterizing, we obtain that $\C_0(f)=\Omega(\UC_1(f)^2)$, which improves over the bound in \eqref{eqn:unambiguous-dnfs-previous-best-separation} by a $\log^6(n)$ factor, and gives an optimal resolution of Puzzle I up to constant factors, given that the upper bound $\C_0(f) \le \UC_1(f)^2$ of \citet{goos2015lower} holds for every $f$.

\subsection{Application: Alon-Saks-Seymour and Clique versus Independent Set}
\label{sec:alon-saks-seymour-clique-independent-set}

As mentioned above, Puzzle I is motivated (among other reasons) by its application to the Alon-Saks-Seymour conjecture and the Clique versus Independent Set problem. The Alon-Saks-Seymour conjecture, first referenced in \citet{kahn1991recent}, is concerned with two combinatorial quantities of a graph $G$: the chromatic number $\chi(G)$ and biclique partition number $\bp(G)$. The chromatic number of $G$ is the smallest number of colors that can be assigned to the vertices in $G$, such that no two adjacent vertices are assigned the same color. The biclique partition number of $G$ is the smallest number of bicliques that the edge set of $G$ can be partitioned into. Alon, Saks and Seymour conjectured that these quantities are related as $\chi(G) \le \bp(G)+1$. While \citet*{GMNS} confirmed this conjecture for all graphs having $\bp(G) \le 9$, in a celebrated result, \citet{huang2012counterexample} disproved it by constructing a family of graphs satisfying $\chi(G) \ge \Omega(\bp(G)^{6/5})$. This separation has since been improved in a series of works, namely: \citet{goos2015lower}, who showed $\chi(G) \ge \exp(\Omega(\log^{1.12}(\bp(G))))$, \citet*{ben2017low}, who showed $\chi(G) \ge \exp(\Omega(\log^{1.22}(\bp(G))))$, and finally \citet{balodis2023unambiguous}, who showed $\chi(G) \ge \exp(\widetilde{\Omega}(\log^{2}(\bp(G))))$. In particular, the precise bound shown by \citet{balodis2023unambiguous} (which is made explicit in \citet[Theorem 1.11]{cheung2023online}\footnote{At first glance, Theorem 1.11 in \citet{cheung2023online} appears to imply the weaker separation $    \chi(G) \ge \exp\left(\Omega\left(\frac{\log^2(\bp(G))}{(\log \log (\bp(G)))^8}\right)\right)$. However, combining their Theorems 3.2 and 3.3 gives an additional $\Omega(\log n)$ factor in $\log\Cov_0$ beyond the bound stated in their Corollary 3.4, yielding the bound shown in \eqref{eqn:alon-saks-seymour-previous-best}.}, and to our knowledge, is the best-known separation so far) was
\begin{align}
    \label{eqn:alon-saks-seymour-previous-best}
    \chi(G) \ge \exp\left(\Omega\left(\frac{\log^2(\bp(G))}{(\log \log (\bp(G)))^7}\right)\right).
\end{align}
Importantly, the bound in \eqref{eqn:alon-saks-seymour-previous-best} is optimal up to the doubly logarithmic factors, since on the upper bound side, it is known that $\chi(G) \le \exp(O(\log^{2}(\bp(G))))$ for all graphs \citep{mubayi2009bipartite,fox2026note}.

The way that \citet{balodis2023unambiguous} obtain the near-optimal separation in \eqref{eqn:alon-saks-seymour-previous-best} is by applying a \textit{lifting theorem} of \citet*{goos2016rectangles} to the family of functions they construct that witnesses $\C_0(f) \ge \widetilde{\Omega}(\UC_1(f)^2)$. The lifting operation converts a Boolean function into a two-party communication problem $h$ by composing through a gadget $g$. The classical result of \citet{goos2016rectangles} constructs a gadget of size $k=\Theta(\log n)$ that translates a lower bound on the $0$-certificate complexity $\C_0(f)$ of any Boolean function $f$ into a lower bound on the 0-monochromatic rectangle covering number $\Cov_0(h)$ of the lifted communication problem $h$ (see \Cref{sec:lifting} for a definition of $\Cov_0(h)$). Directly applying this all-purpose lifting theorem to our optimal unambiguous DNFs constructed in \Cref{thm:unambigous-dnf} improves the bound in \eqref{eqn:alon-saks-seymour-previous-best} by a $(\log \log (\bp(G)))^6$ factor in the exponent to $ \chi(G) \ge \exp\left(\Omega\left(\frac{\log^2(\bp(G))}{\log \log (\bp(G))}\right)\right)$. However, this still leaves a $\log \log (\bp(G))$ factor gap to the upper bound.

While the lifting theorem of \cite{goos2016rectangles} is extremely powerful, in that it applies to \textit{every} Boolean function, the fact that it uses a gadget of size $k=\Theta(\log n)$ is precisely what leads to the superfluous doubly logarithmic factor in the final separation between chromatic number and biclique partition number. One way to shave this last remaining factor is to derive a lifting theorem that uses only a \textit{constant-sized gadget}, while retaining the guarantee of the lifting theorem of \citet{goos2016rectangles} --- not necessarily for all Boolean functions, but even just for the unambiguous DNFs from \Cref{thm:unambigous-dnf}. Such constant-gadget lifting theorems are notoriously hard to come by; see for example the recent work by \citet{alekseev2025lifting}, which provides a detailed examination of lifting theorems.

Nevertheless, by exploiting the special structure of our unambiguous DNFs, we are able to derive precisely the required constant-gadget lifting theorem.

\begin{restatable}[Constant-gadget Lifting]{theorem}{Lifting}
	\label{thm:constant-gadget-lifting}
    There exists a gadget function $g:\{0,1\}^{3} \times \{0,1\}^{3} \to \{0,1\}$, such that for all large enough $n$, for the DNF $f$ constructed in \Cref{thm:unambigous-dnf}, denoting by $h$ the lifted communication problem $f \circ g^{n^2}: \{0,1\}^{3n^2} \times \{0,1\}^{3n^2} \to \{0,1\}$, it holds that
    \begin{align*}
        \log \Cov_0(h) = \Omega(n^2).
    \end{align*}
\end{restatable}
The gadget $g$ that we use in the lifting theorem above is a simple cyclic gadget 
$$g(x,y) = 0 \iff y-x \in \{-1,0,1\} \mod 8,$$ 
where we identify $x,y \in \{0,1\}^3$ by elements in $\Z_8$ in the natural way. We defer a detailed discussion of the gadget to the overview of results (\Cref{sec:overview}).

Applying this lifting theorem to our unambiguous DNFs, and following the rest of the reduction given in \citet{cheung2023online}, we obtain an optimal refutation of the Alon-Saks-Seymour conjecture.

\begin{restatable}[Optimal Refutation of Alon-Saks-Seymour]{theorem}{AlonSaksSeymour}
    \label{thm:alon-saks-seymour}
    For infinitely many $n$, there exists a graph $G$ having $2^{\Theta(n^2)}$ vertices that admits a biclique partition of size $2^{O(n)}$ but has chromatic number $2^{\Omega(n^2)}$. Consequently, there exists an infinite family of graphs $G$ satisfying
    \begin{align}
        \label{eqn:alon-saks-seymour-improved}
        \chi(G) \ge \exp\left(\Omega\left(\log^2(\bp(G))\right)\right).
    \end{align}
\end{restatable}

This matches the upper bounds of \citet{mubayi2009bipartite,fox2026note} up to constants in the exponent.

Notably, the theorem above also \textit{uniformly} improves all the parameters from Theorem 1.11 in \cite{cheung2023online}. In particular, the form of their result is similar to \Cref{thm:alon-saks-seymour} above, but in addition to a weaker separation, the number of vertices in their graph is $2^{\Theta(n^4\log^3n)}$, while our result only requires $2^{\Theta(n^2)}$ vertices. Observe that any graph having chromatic number $\chi$ has at least $\chi$ vertices. This means that the size of our graph in \Cref{thm:alon-saks-seymour} which witnesses the optimal refutation of the Alon-Saks-Seymour conjecture is \textit{itself} optimal up to constant factors in the exponent. We find it remarkable that an optimal separation between the chromatic number and biclique partition number can be witnessed by graphs that effectively have the smallest number of vertices required to witness the separation!

We next describe the Clique versus Independent Set problem \citep{yannakakis1988expressing,bousquet2014clique}. This is a two-party communication problem, wherein there is an underlying graph $G=(V,E)$ on $n$ vertices, Alice is given a clique $x \subseteq V$, Bob gets an independent set $y \subseteq V$, and their goal is output $\CIS_G(x,y) :=\Ind[x \cap y \neq \emptyset]$. \citet{yannakakis1988expressing} showed that the co-nondeterministic communication complexity of this problem is $O(\log^2(n))$.

With respect to lower bounds, a lower bound of $\Omega(\log(n))$ on the co-nondeterministic communication complexity was obtained in the works of \citet{huang2012counterexample,amano2014some,shigeta2015ordered}. This was later improved to $\Omega(\log^{1.12}(n))$ in \citet{goos2015lower}, $\Omega(\log^{1.22}(n))$ in \citet{ben2017low}, and finally to a near-optimal $\widetilde{\Omega}(\log^2(n))$ in \citet{balodis2023unambiguous}. In fact, the Alon-Saks-Seymour problem and the Clique versus Independent Set problem are deeply related \citep{bousquet2014clique}: a separation $\chi(G) \ge \exp\left(\Omega\left(\frac{(\log(\bp(G)))^a}{(\log \log (\bp(G)))^b}\right)\right)$ implies a co-nondeterministic lower bound $\Omega\left(\frac{(\log(\bp(G)))^a}{(\log \log (\bp(G)))^b}\right)$ for some $\CIS_G$ on at most $\bp(G)$ vertices. So, keeping \eqref{eqn:alon-saks-seymour-previous-best} in mind, the precise form of the lower bound of \citet{balodis2023unambiguous} is $\Omega\left(\frac{\log^2(n)}{(\log \log (n))^7}\right)$.

As a direct implication of \eqref{eqn:alon-saks-seymour-improved}, we obtain an optimal $\Omega(\log^2(n))$ lower bound for the Clique versus Independent Set problem, matching the upper bound of \citet{yannakakis1988expressing} to constant factors.
\begin{corollary}[Clique versus Independent Set]
    \label{corollary:clique-independent-set}
    There exists a family of $\CIS_G$ on $n$ vertices for which the co-nondeterministic communication complexity of $\CIS_G$ is at least $\Omega\left(\log^2(n)\right)$.
\end{corollary}

Finally, we note that our constant-gadget lifting theorem also helps shave all the log factors from the quadratic lower bound of \citet{goos2018deterministic} for the log-rank conjecture; see \Cref{corollary:log-rank}.

\subsection{Application: Query Complexity}
\label{sec:query-complexity}

We now describe the implications of our optimal unambiguous DNF separation from \Cref{thm:unambigous-dnf} to query complexity. In particular, we derive improved separations between the certificate complexity of a Boolean function and other standard complexity measures.

We start with the most significant result: an optimal separation between the certificate complexity and \textit{approximate degree} of Boolean functions. The $\eps$-approximate degree $\widetilde{\Deg}_\eps(f)$ of a Boolean function on $n$ variables is the smallest degree of an $n$-variate polynomial $p:\R^n \to \R$ such that $p(x) \in f(x) \pm \eps$ for all $x \in \{0,1\}^n$. Define $\widetilde{\Deg}(f) := \widetilde{\Deg}_{1/3}(f)$. It is known that $\C(f) \le \widetilde{O}(\widetilde{\Deg}(f)^4)$ for every $f$ (see Table 1 in \cite{aaronson2021degree}), whereas \cite{balodis2023unambiguous} showed that there exists a family of functions $f$ for which $\C(f) \ge \widetilde{\Omega}(\widetilde{\Deg}(f)^3)$. We close the gap between the upper and lower bounds (up to polylogarithmic factors).
\begin{restatable}[Certificate Complexity vs.\ Approximate Degree]{theorem}{ApproximateDegree}
    \label{thm:approximate-degree}
    There exists a family of Boolean functions $f:\{0,1\}^{O(n^3 \log n)} \to \{0,1\}$ with $\C(f) \ge \Omega\left(\frac{\widetilde{\Deg}(f)^4}{\log^6(n)}\right)$.
\end{restatable}

We next state the more direct corollaries of our \Cref{thm:unambigous-dnf}. Similar to approximate degree, define the \textit{exact} degree $\Deg(f)$ of a Boolean function $f$ to be the smallest degree of a polynomial that equals $f$ at all Boolean inputs. It is known that $\C(f) \le \widetilde{O}(\Deg(f)^3)$ for all $f$ \citep{aaronson2021degree}. On the other hand, it turns out that $\Deg(f) \le \UC_1(f)$, and thus a lower bound of $\C(f) \ge \Omega\left(\frac{\Deg(f)^2}{\log^6(n)}\right)$ follows directly as a result of the separation for unambiguous DNFs \eqref{eqn:unambiguous-dnfs-previous-best-separation} shown by \citet{balodis2023unambiguous}. Since our separation for unambiguous DNFs in \Cref{thm:unambigous-dnf} is log-free, we obtain:

\begin{corollary}[Certificate Complexity vs.\ Degree]
    \label{corollary:degree}
    There exists a family of Boolean functions $f$ with $\C(f) \ge \Omega(\Deg(f)^2)$.
\end{corollary}

Our last corollary has to do with the \textit{sensitivity} of a function. The sensitivity of $x \in \{0,1\}^n$ with respect to an $n$-bit Boolean function $f$ is the number of bits in $x$, which when flipped in isolation, flip the value of $f(x)$. The sensitivity of $f$, denoted $\s(f)$, is the maximum sensitivity of any $x$ with respect to $f$. The result of \cite{balodis2023unambiguous} shows the existence of a family of functions $f$ for which $\C(f) \ge \Omega\left(\frac{\s(f)^3}{\log^6(n)}\right)$, whereas there is a global upper bound $\C(f) \le \widetilde{O}(\s(f)^5)$ \citep{aaronson2021degree}. Our result shaves the polylogarithmic factors in the lower bound of \citet{balodis2023unambiguous} (proof in \Cref{sec:sensitivity-proof}).
\begin{restatable}[Certificate Complexity vs.\ Sensitivity]{corollary}{Sensitivity}
    \label{corollary:sensitivity}
    There exists a family of Boolean functions $f$ with $\C(f) \ge \Omega(\s(f)^3)$.
\end{restatable}

\subsection{Application: Multiclass Sample Compression}
\label{sec:multiclass-sample-compression-intro}

Our final application is the problem of sample compression in learning theory \citep{littlestone1986relating}. Consider a class $\mcC$ of functions/concepts mapping a domain to a discrete label space of $c$ labels. The sample compression problem asks for a pair of maps $(\kappa, \rho)$, where $\kappa$ is the \textit{compressor} and $\rho$ is the \textit{reconstructor}. The compressor compresses a labeled sample $S$ to a subsample $S'$, and the reconstructor reconstructs labels on $S \setminus S'$ from the compressed sample $S'$. Relevant here is also a notion of complexity of the class $\mcC$ known as the \textit{Natarajan} dimension \citep{natarajan1989learning}, which governs the learnability of $\mcC$. It is known that for any concept class having Natarajan dimension $d$, there exists a sample compression scheme, such that for any sample labeled by some concept in the class, it is possible to reconstruct labels on the entire sample from a compressed sample of size $c^{O(d)}$; moreover, the size of the compressed sample provably needs to be at least $\Omega(d)$ (e.g., see \cite{pabbaraju2024multiclass}).

Note that if we keep the Natarajan dimension fixed, the lower bound on the sample compression size is constant, and does not grow with the number of labels. \cite{pabbaraju2024multiclass} asked as an open question if the lower bound on sample compression must grow with $c$, especially given that the upper bound in this regime is polynomial in $c$. Using the optimal small-size refutation of the Alon-Saks-Seymour conjecture given by \Cref{thm:alon-saks-seymour}, we provide a partial resolution to this open question, by constructing a family of classes for which the size of any sample compression scheme must grow at least polylogarithmically with the number of labels.

\begin{restatable}[Sample Compression Lower Bound for Multiclass Concept Classes]{theorem}{Compression}
    \label{thm:multiclass-sample-compression-lower-bound}
    There is a family of multiclass concept classes having Natarajan dimension 1, for which the size of any valid sample compression scheme as a function of the number of labels $c$ is at least $\Omega\left(\sqrt{\log c}\right)$. 
\end{restatable}

We remark that \Cref{thm:alon-saks-seymour} also shaves the lower-order terms from the results of \citet{alon2022theory} and \citet{pabbaraju2024multiclass} that depend on the Alon-Saks-Seymour refutation (i.e., the $1-o(1)$ exponent in Theorems 7 and 11 in \citep{alon2022theory} and Theorem 1 in \citet{pabbaraju2024multiclass}).

\subsection{Optimal Solutions to Other Puzzles}
\label{sec:other-puzzles-optimality}

Lastly, we note that our optimal unambiguous DNF construction also leads to optimal solutions to Puzzles II and III described in \citet{balodis2023unambiguous}, both of which were shown by them to be equivalent to Puzzle I up to polylogarithmic factors. %
These puzzles are framed in the language of partial functions and intersecting hypergraphs.

\paragraph{Partial Functions.} The notion of certificates naturally extends to \textit{partial functions} $f$, where $f(x)$ may be undefined at some values of $x$ (i.e., $f(x)=*$). Namely, a (Boolean) partial function is a map $f:\{0,1\}^n \to \{0,1,*\}$. In addition to the $0$ and $1$ certificate complexities of $f$ defined above, we can also define the $\overline{0}$ (read ``not-0'') and $\overline{1}$ (read ``not-1'') certificate complexities of $f$. Namely, a partial input $\rho \in \{0,1,*\}^n$ is a $\overline{0}$-certificate for $f$ at $x$, if $\rho$ is consistent with $x$, and for every $x'$ consistent with $\rho$, it holds that $f(x') \neq 0$. The $\overline{0}$-certificate complexity of $f$ at $x$, denoted $\C_{\overline{0}}(f, x)$, is the size of the smallest $\overline{0}$-certificate for $x$. The $\overline{0}$-certificate complexity of $f$, denoted $\C_{\overline{0}}(f)$, is defined as $\max_{x \in \{0,1\}^n\,:\, f(x) \neq 0} \C_{\overline{0}}(f, x)$. We similarly define the $\overline{1}$-certificate complexity $\C_{\overline{1}}(f)$. 

The second puzzle studied in \citet{balodis2023unambiguous} pertains to the separation between $\C_{\overline{0}}(f), \C_{\overline{1}}(f)$ and $\C(f)$:

\paragraph{Puzzle II.} For $\alpha > 1$, does there exist a family of partial functions $f$ together with an $x \in f^{-1}(*)$ such that both $\C_{\overline{0}}(f, x)$ and $\C_{\overline{1}}(f, x)$ are at least $\Omega(\C(f)^\alpha)$?
\\

\citet{balodis2023unambiguous} constructed a family of partial functions $f$, for which the separation above is witnessed up to polylogarithmic factors with $\alpha=2$. Concretely, they showed
\begin{align}
    \label{eqn:partial-functions-previous-best-separation}
    \min\left\{\C_{\overline{0}}(f, x), \C_{\overline{1}}(f, x)\right\} \ge \Omega\left(\frac{\C(f)^2}{\log^2(n)}\right).
\end{align}
Using the implication Puzzle I $\to$ Puzzle II  given in \cite[Section 5.2]{balodis2023unambiguous}, together with \Cref{thm:unambigous-dnf} above, we get:
\begin{corollary}[Partial Functions]
    \label{corollary:partial-functions}
    For any $n \ge 2$, there exists a partial function $f:\{0,1\}^{2n^2} \to \{0,1,*\}$, such that there exists an input $x \in f^{-1}(*)$ satisfying
    \begin{align*}
        \min\left\{\C_{\overline{0}}(f, x), \C_{\overline{1}}(f, x)\right\} \ge \Omega(n^2), \qquad \C(f) \le O(n).
    \end{align*}
\end{corollary}
Reparameterizing, we obtain that $\min\left\{\C_{\overline{0}}(f, x), \C_{\overline{1}}(f, x)\right\} \ge \Omega(\C(f)^2)$, which improves over the bound in \eqref{eqn:partial-functions-previous-best-separation} by a $\log^2(n)$ factor. While a quadratic separation here is optimal due to the equivalence of the three puzzles, the reductions between the puzzles obtained in \citet{balodis2023unambiguous} lose polylogarithmic factors; thus, it is a priori not clear that a separation better than the one in \Cref{corollary:partial-functions} by polylogarithmic factors cannot exist. Nevertheless, in \Cref{sec:optimality-partial-functions-intersecting-hypergraphs}, we show that a direct upper bound $\min\left\{\C_{\overline{0}}(f, x), \C_{\overline{1}}(f, x)\right\} \le \C(f)^2$ always holds. Thus, our result above is an optimal resolution of Puzzle II up to constant factors. 

\paragraph{Intersecting Hypergraphs.} The third and final puzzle from \citet{balodis2023unambiguous} is purely graph-theoretic. First, the definitions again: a hypergraph $G=(V,E)$ is said to be \textit{intersecting} if $e \cap e' \neq \emptyset$ for every $e, e' \in E$. The rank $r(G)$ of a hypergraph is the size of its largest edge. A subset $U \subseteq V$ is said to be $c$-monochromatic with respect to a coloring function $c:V \to \{0,1\}$ if $c(u)$ is the same for every $u \in U$. Finally, a set $S \subseteq V$ is a \textit{hitting set} for $E$ if $S \cap e \neq \emptyset$ for every $e \in E$.

\paragraph{Puzzle III.} For $\alpha > 1$, does there exist a family of intersecting hypergraphs $G=(V,E)$ together with a coloring $c:V \to \{0,1\}$ such that every $c$-monochromatic hitting set has size at least $\Omega(r(G)^\alpha)$?
\\

For this problem, \citet{balodis2023unambiguous} constructed a family of intersecting hypergraphs $G$ and coloring functions $c$ such that %
every $c$-monochromatic hitting set has size at least 
\begin{align}
    \label{eqn:intersecting-hypergraphs-previous-best-separation}
    \Omega\left(\frac{r(G)^2}{\log^2(n)}\right).
\end{align}
Using the implication Puzzle II $\to$ Puzzle III given in \cite[Section 5.3]{balodis2023unambiguous}, together with \Cref{corollary:partial-functions} above, we also have the following:
\begin{corollary}[Intersecting Hypergraphs]
    \label{corollary:intersecting-hypergraphs}
    For any $n \ge 2$, there exists an intersecting hypergraph $G=(V, E)$ with $|V|=4n^2+2$, $r(G)=O(n)$ and a coloring $c:V \to \{0,1\}$ such that every $c$-monochromatic hitting set has size at least $\Omega(n^2)$.
\end{corollary}
By converting parameters, we obtain a family of intersecting hypergraphs $G$ and colorings $c$ for which every $c$-monochromatic hitting set has size at least $\Omega(r(G)^2)$, improving the bound in \eqref{eqn:intersecting-hypergraphs-previous-best-separation} by a $\log^2(n)$ factor. Again, we show that this separation is optimal up to constant factors, by giving a direct matching upper bound in \Cref{sec:optimality-partial-functions-intersecting-hypergraphs}. %

\section{Overview of Techniques}
\label{sec:overview}

In this section, we sketch the main technical ideas used in the construction of unambiguous DNFs in \Cref{thm:unambigous-dnf}, the constant-gadget lifting in \Cref{thm:constant-gadget-lifting}, the quartic separation between certificate complexity and approximate degree in \Cref{thm:approximate-degree}, and the sample compression lower bound in \Cref{thm:multiclass-sample-compression-lower-bound}; the detailed proofs of these results are given in \Cref{sec:unambigous-dnf-separation,sec:lifting,sec:approx-degree-quartic-separation,sec:multiclass-sample-compression}. %

\paragraph{Unambiguous DNFs.} Our unambiguous DNFs arise from a simple yet carefully constructed set-pair system $\{(P_T,N_T)\}_T$ over a universe $U$, where each set-pair $(P_T,N_T)$ comprises a positive set $P_T \subseteq U$ and a negative set $N_T \subseteq U$. Each pair will correspond to a term $C_T$ in the DNF of the form
\begin{align}
    \label{eqn:term-form}
    C_T = \bigwedge_{u \in P_{T}} x_u \wedge \bigwedge_{v \in N_{T}} \neg x_{v},
\end{align}
and the final DNF will be $f=\bigvee_T C_T$. Recall that unambiguity requires that only one term be satisfied by any $x$ for which $f(x)=1$. The crucial advantage offered by formulating the DNF as a set-pair system is that we can formulate a sufficient condition for unambiguity in terms of a \textit{pairwise incompatibility} constraint on the set-pair system. Namely, if we ensure that for every distinct pair $(P_T, N_T)$ and $(P_{T'}, N_{T'})$, it holds that either $P_T \cap N_{T'} \neq \emptyset$ or $P_{T'} \cap N_{T} \neq \emptyset$, then whenever a term $(P_T, N_T)$ is satisfied, it is ensured that every other term $(P_{T'}, N_{T'})$ is unsatisfied.

So, we will aim to build a set-pair system that has such pairwise incompatibility; in addition, we want to ensure a quadratic separation between the 0-certificate complexity of $f$ and the term-width of $f$. We first consider the $0$-certificate complexity. Note that by the form \eqref{eqn:term-form} of a term in our DNF, $f(0^{|U|})=0$, since the positive part of every term is falsified. Let us study the $0$-certificate complexity of $0^{|U|}$. For any certificate $\rho$ consistent with $0^{|U|}$, observe that if $\dom(\rho)$ does not intersect every $P_T$ in our set-pair system, then we can take a $P_T$ that is disjoint with $\dom(\rho)$, and set all the variables in it to 1; this extends $\rho$ to an input $x$ that is consistent with $\rho$ but $f(x)=1$, which invalidates $\rho$ from being a $0$-certificate for $0^{|U|}$. Thus, any 0-certificate $\rho$ for $0^{|U|}$ must satisfy that $\dom(\rho)$ is a \textit{hitting set} for the family of all positive sets, meaning that the $0$-certificate complexity of $f$ is at least the hitting number of the family of positive sets.

A simple choice of the positive sets which ensures a large hitting number is as follows: consider the universe $U=B_1 \sqcup B_2 \sqcup \dots \sqcup B_n$ to be the disjoint union of buckets $B_i$, where each $|B_i|=n$, and hence $|U|=n^2$. We include every subset of every $B_i$ that has size $n/2$ as a positive set. That is, for every $T=(i,S)$, where $i \in [n]$ and $S \subseteq B_i, |S|=n/2$, we include a positive set $P_{T}=P_{i,S}=S$. Then, any hitting set must include at least $n/2$ elements of every $B_i$, giving us a hitting number of $\Omega(n^2)$. Our remaining goal is to suitably choose the negative sets $N_T$, so that the term-width of every term, namely $|P_T|+|N_T|$ is at most $O(n)$, while also satisfying the pairwise incompatibility condition to ensure unambiguity.

Our negative sets $N_T=N_{i,S}$ will have two components: the set $B_i \setminus P_{i,S}$ (i.e., the complement of the positive term within the bucket $B_i$), together with a cross-bucket term $A_{i,j}(S)$ for every $j \neq i$. The first component ensures incompatibility with all the other terms that arise from the same bucket, whereas the second component ensures incompatibility with terms arising from different buckets. It remains to choose the sets $A_{i,j}(S)$ appropriately such that $\sum_{j \neq i}|A_{i,j}(S)|=O(n)$ for every $i$ and $S$. We use the probabilistic method for constructing these sets. For every (unordered) pair $\{i,j\}$, we randomly hash the buckets $B_i$ and $B_j$ independently to $[n]$. We include in $N_{i,S}$ precisely those terms in bucket $B_j$ that get hashed to a value which is at most the \textit{smallest} value that an element from $S$ gets hashed to. It is not too hard to verify that this guarantees the pairwise incompatibility property. Additionally, since $|S| = n/2$, the distribution of the smallest value that an element from $S$ gets hashed to has an exponentially small tail. We can thus union bound an exponentially small failure probability over all choices of $i,S$, and conclude that $\sum_{j \neq i}|A_{i,j}(S)|=O(n)$ holds for every $i,S$ with positive probability. This gives the desired set-pair system, and hence an unambiguous DNF that has a quadratic separation between term-width and 0-certificate complexity.

\paragraph{Constant-gadget Lifting.} We will now sketch how, by exploiting the special structure of the unambiguous DNF constructed above, we are able to lift it to a communication problem with a constant-sized gadget. For the DNF $f:\{0,1\}^{|U|} \to \{0,1\}$ above, where $|U|=n^2$, our goal is to construct a gadget $g: \{0,1\}^k \times\{0,1\}^k \to \{0,1\}$ where $k=O(1)$, such that the communication problem $h = f \circ g^{|U|} :\{0,1\}^{k|U|} \times \{0,1\}^{k|U|} \to \{0,1\}$ satisfies $\log\Cov_0(h) \ge \Omega(|U|)$. Here, $\Cov_0(h)$ denotes the $0$-monochromatic covering number of $h$, i.e., the minimum number of $0$-monochromatic rectangles required to cover all the $0$-entries in the communication matrix of $h$,\footnote{The communication matrix of $h: \mcX \times \mcY \to \{0,1\}$ is a Boolean matrix $A^{|\mcX| \times |\mcY|}$ with $A_{x,y}=h(x,y)$. For $b \in \{0,1\}$, a $b$-monochromatic rectangle in $A$ is a set $R=X \times Y$, where $X \subseteq \mcX, Y \subseteq \mcY$ such that $h(x,y)=b$ for all $(x,y)\in R$.} and
\begin{align*}
    f \circ g^{|U|}((x_1,\dots,x_{|U|}), (y_1,\dots,y_{|U|})) = f(g(x_1,y_1),\dots,g(x_{|U|}, y_{|U|})),
\end{align*}
where each $x_i, y_i \in \{0,1\}^k$.

We first recall that our DNF satisfies $f(0^{|U|})=0$. So, if the gadget $g$ is such that $g(x,x)=0$ for all $x$, then we get that the entire diagonal of the communication matrix of $h$ is $0$. This is useful, since any set of $0$-monochromatic rectangles that cover all the $0$-entries in the matrix must necessarily cover all the diagonal entries. Our objective then reduces to showing that any individual rectangle may cover only a small number of the diagonal entries.

Towards this, suppose we could construct a non-negative matrix $M^{N \times N}$, where $N := 2^{k|U|}$, which has the property that
\begin{align*}
    M(x,y) > 0 \implies h(x,y)=1.
\end{align*}
(The communication matrix of $h$ itself has this property, but we later end up requiring additional properties of the matrix $M$ that are not directly satisfied by the communication matrix itself.) Suppose also that every row in $M$ has the same sum $d > 0$ --- we might expect such a property to naturally hold if we choose a suitably symmetric gadget. Now let us interpret $M$ as the adjacency matrix of an undirected graph $G_M$ over $N$ vertices, where $x$ and $y$ are connected by an edge iff $M(x,y) > 0$. Fix any $0$-cover $\mcF$ of $h$ by $0$-monochromatic rectangles, and let $R$ be a rectangle in this cover. Now consider the set $S_R=\{x: (x,x) \in R\}$ of diagonal entries in $R$ --- from the preceding paragraph, it must hold that $N \le \sum_R |S_R| \le |\mcF| \max_{R}|S_R|$. But observe that by virtue of $R$ being a $0$-monochromatic rectangle, for any $x \neq y \in S_R$, it holds that $h(x,y)=0$, meaning that $x$ and $y$ are \textit{not} connected by an edge in the graph $G_M$. Thus, the set $S_R$ is an \textit{independent set} in $G_M$. We then aim to show that no independent set in $G_M$ is too large.

For any independent set $S$ in $G_M$, let $\bone_S$ be its indicator vector. Since the diagonal of $M$ is 0 and $S$ is an independent set, it holds that $\bone_S^T M \bone_S=0$. Consider decomposing $\bone_S$ along the along-ones vector $\bone$ and orthogonal to it as $\bone_S = \delta \bone + z$, where $\delta = |S|/N$. A routine calculation shows that 
\begin{align*}
    \|z\|_2^2 = \delta N (1-\delta), \qquad z^TMz=-\delta^2Nd.
\end{align*}
But by the variational form of the minimum eigenvalue $\lambda_{\min}$ of $M$, we also have that $z^TMz \ge \lambda_{\min} \|z\|_2^2 = \lambda_{\min} \delta N (1-\delta)$. Thus, if we could show that the minimum eigenvalue $\lambda_{\min}$ of $M$ is at least $-d\eps$ for $\eps > 0$, we would obtain an upper bound of $\eps N$ on the size of any independent set $S$. Combining with the calculation from the previous paragraph, this would lower bound $|\mcF|$ as
\begin{align*}
    N \le |\mcF|\eps N \implies |\mcF| \ge 1/\eps.
\end{align*}
In particular, the value of $\eps$ we will aim for is $\eps=2^{-\Omega(n^2)}$, which will give us that $\log|\mcF|=\Omega(n^2)$ --- our desired lower bound.

Summarily, our entire task has been distilled into a concrete objective: construct a constant-sized gadget $g$ and associated matrix $M$ which satisfies that $M(x,y) > 0 \implies h(x,y)=1$; furthermore, $M$ should be non-negative, symmetric, have a common row-sum $d > 0$, and also have the desired spectral property that its minimum eigenvalue $\lambda_{\min}$ is not too negative --- namely at least $-d\cdot 2^{-\Omega(n^2)}$. To this end, we first tried several natural gadgets like the AND, XOR and AND/OR gadgets, but to no avail. Nevertheless, it turns out that a simple \textit{cyclic gadget} $g:\{0,1\}^3 \times \{0,1\}^3 \to \{0,1\}$ over 3 bits works. Concretely, by identifying every $x \in \{0,1\}^3$ as the binary representation of an element in $\Z_8=\{0,1,\dots,7\}$ in the natural way, define
\begin{align}
    \label{eqn:overview-gadget-def}
	g(x,y) = \begin{cases}
		0 & \text{if $y-x \in \{-1,0,1\}\mod 8$}, \\
		1 & \text{otherwise.}
	\end{cases}
\end{align}
That is, $g$ outputs $0$ only when $y$ is either equal to $x$, or equal to one of the two neighbors of $x$ around the cycle $0,1,\dots,7,0$. Note that the gadget is symmetric, namely $g(x,y)=g(y,x)$.

While there has been precedent for using such constant-sized modular gadgets in the literature --- e.g., the \textit{versatile} gadget $g(x,y)= 1 \iff x+y \in \{2,3\} \mod 4$ \citep{alexander2011pattern,goos2014communication,anshu2020query,goos2022communication} %
--- to the best of our knowledge, the precise cyclic gadget above has not been used specifically in the context of query-to-communication lifting. In particular, the matrix of the versatile gadget  has an equal number of 0's and 1's (namely 2 each) in every row, whereas we do end up crucially using that the cyclic gadget has an unequal number (3 and 5) in our spectral analysis. Another reason to use the cyclic gadget is that the precise spectrum of the associated Cayley graph is well-understood \citep{trevisannotes,nica2018brief,larsen2026sample}.

In slightly more detail, we obtain the desired matrix $M$ as follows: we first define two simple $8 \times 8$ matrices $W_0$ and $W_1$ that satisfy 
\begin{align*}
    W_b(x,y) > 0 \iff g(x,y) = b
\end{align*}
for the cyclic gadget $g$. %
Then, for any $i \in [n]$ and any $S \subseteq B_i, |S|=n/2$, with $\bone_S$ being the indicator vector for the set $S$, consider defining the matrix $M_{i,S}$ as the Kronecker product
\begin{align}
    \label{eqn:overview-M-i-S}
    M_{i,S} = \bigotimes_{u \in U} W_{\bone_{S,u}} \quad \text{ where } \quad M_{i,S}(x,y) = \prod_{u \in U} W_{\bone_{S,u}}(x_u, y_u).
\end{align}
Here, $\bone_{S,u}=\Ind[u \in S]$, and every $x_u, y_u \in \Z_8$. That is, $M_{i,S}$ has a ``component'' $W_1$ at the index $u$ in the Kronecker product if $u \in S$, and $W_0$ otherwise. Finally, define
\begin{align*}
	M = \sum_{i=1}^n \sum_{S \in \binom{B_i}{n/2}} M_{i,S}.
\end{align*}
It can be readily verified that $M$ is non-negative, symmetric and has a common row-sum $d > 0$. Furthermore, observe that if $M(x,y) > 0$, then it must be the case that $M_{i,S}(x,y) > 0$ for at least some $(i,S)$, meaning that every term in the product in \eqref{eqn:overview-M-i-S} is positive. By definition, this means that $(g(x_u, y_u)_{u \in U})=\bone_S$, which means $h((g(x_u, y_u))_{u \in U})=f(\bone_S)=1$, as desired.

All that remains is to argue that the minimum eigenvalue of $M$ is at least $-d\cdot 2^{-\Omega(n^2)}$. As mentioned above, we can compute precise expressions for all the eigenvalues of the matrix $M$ since it is derived from the cyclic gadget. Concretely, every eigenvalue can be written (up to a multiplicative factor) as an \textit{elementary symmetric polynomial} of degree $n/2$. Bounding these symmetric polynomials ends up being the most technically involved part of the paper. To avoid overwhelming the reader further in this overview, we provide intuition for this analysis in the sections containing it, rather than delving into it here.

\paragraph{Quartic Separation between Certificate Complexity and Approximate Degree.} The proof of \Cref{thm:approximate-degree} has elements inspired from the \textit{cheat-sheet} framework of \cite{aaronson2016separations}, which was also used by \citet{balodis2023unambiguous} to establish their cubic separation between certificate complexity and approximate degree. %
The high-level idea is to package the unambiguous DNF $f$ constructed in \Cref{thm:unambigous-dnf} into another function $G$, such that the lower bound of $\Omega(n^2)$ on 0-certificate complexity carries over, but the resulting function is easy-to-approximate with a polynomial having low degree. The mechanism we use to achieve the latter property is to have part of the input in the packaged function be the input $x$ to $f$ itself, and have the rest of the input store several candidate \textit{certificates} for certifying $f$ at $x$. Using the set-pair-based structure of the unambiguous DNF $f$, we then design a verifiably sound and complete certification scheme for certifying the input $x$ to be a 1-input. The verification can furthermore be achieved by performing only a few ``local checks'', which allows for it to be approximated by a low-degree polynomial. In particular, we express the verification as an AND of only $O(n)$-many (deterministic) decision tree evaluations having depth $O(\log n)$. Since every decision tree can be exactly expressed by a polynomial having degree on the order of the depth of the decision tree, and the AND function on $O(n)$ inputs can also be approximated by a polynomial of degree $\widetilde{O}(\sqrt{n})$ \citep{nisan1994degree,de2008note}, the entire verification can be approximated by a polynomial of degree $\widetilde{O}(\sqrt{n})$.

In order to transfer the $\Omega(n^2)$ 0-certificate complexity to the packaged function, we use $\Omega(n^2)$ many candidate certificate slots, and further rig up the certification process (using the unambiguity of $f$) to satisfy that at most one candidate certificate may \textit{validly} certify $f$. Our final packaged function $G$ is then the indicator that at least one of the candidate certificates validly certifies $f$. Intuitively, if one now only fixes $< n^2$ bits in a 0-input to $G$, this leaves enough room to fill up part of the input to $f$, as well as a complete certificate slot with a valid certificate, so as to certify $f$ to evaluate to 1. This ensures that $\C(G) \ge \Omega(n^2)$. On the other hand, verification of certificates is still approximated by  a polynomial of degree $\widetilde{O}(\sqrt{n})$, using the reasoning from the previous paragraph, and the property that only one slot may validly certify $f$. Combining these yields the claimed quartic separation.

It is essential that the AND above is only over $O(n)$ many local checks (so that it may be approximated by an $\widetilde{O}(\sqrt{n})$ degree polynomial) for the final quartic separation that we desire, and the special structure of our unambiguous DNFs turns out to be crucial for this property. To belabor the point, one could also have attempted a similar construction as above, but instead basing it on the unambiguous DNFs derived in \citet{balodis2023unambiguous}. Doing so still results in a function that has large $0$-certificate complexity, but it is not so clear how one would encode a certificate description that may be verified (soundly and completely) by performing only $O(n)$-many local checks. This is also precisely why \citet{balodis2023unambiguous} use a separate construction based on the Hex board game for their separation between certificate complexity and approximate degree; while this Hex-based construction does indeed allow for a certificate description that can be verified with $O(n)$ local checks, its 0-certificate complexity ends up being smaller (only $\Omega(n^{1.5})$), which eventually leads to their weaker cubic separation.

\paragraph{Sample Compression.} Our proof of \Cref{thm:multiclass-sample-compression-lower-bound} is largely inspired by the sample compression lower bound for multiclass concept classes shown in \cite{pabbaraju2024multiclass}, which translates an Alon-Saks-Seymour refutation into a multiclass concept class that is hard to compress. The strategy is as follows: first, a partial concept class (whose range is $\{0,1,*\}$) having VC dimension 1 is constructed from a graph witnessing a separation between chromatic number and biclique partition number. The class has a partial concept corresponding to every vertex in the graph. Then, a total multiclass concept class that \textit{disambiguates} the partial concept class is constructed by having a total concept for every partial concept, where the total concept fills in the $*$ values with a unique label that is distinct from all other total concepts. Because the labels used for disambiguation are unique, the Natarajan dimension of this total concept class is equal to the VC dimension of the partial class. Since there is one disambiguating total concept that uses a unique label per partial concept in the class, and the number of partial concepts is equal to the number of vertices in the graph, the number of labels in the disambiguating class is on the order of the number of vertices in the graph. We thus apply this strategy to our small-sized graph %
from \Cref{thm:alon-saks-seymour}. Finally, we use the observation that a sample compression scheme for the total class implies a coloring of the vertices of the graph, where the number of colors scales with the size of the sample compression; because the chromatic number of the graph is lower-bounded in terms of the number of labels, we infer the desired lower bound.

We note that applying the construction above to the graph given by Theorem 1.11 in \cite{cheung2023online} would already have given a weaker sample compression lower bound of $\widetilde{\Omega}((\log c)^{1/4})$; our improvement on the size of the graph in \Cref{thm:alon-saks-seymour} is crucial for the stronger lower bound of \Cref{thm:multiclass-sample-compression-lower-bound}.

\section{Optimal Unambiguous DNF Separation}
\label{sec:unambigous-dnf-separation}

In this section, we restate and give a complete proof of \Cref{thm:unambigous-dnf}.
\UnambiguousDNF*
\begin{proof}
    For each $n$, we will derive the DNF $f$ from a ``signed'' set-pair system $\mcT = \{(P_T, N_T)\}_T$ that has several convenient properties.

    \paragraph{Signed set-pair system.} Fix the universe $U$ to be
    \begin{align*}
        U = B_1 \sqcup B_2 \sqcup \dots \sqcup B_n,
    \end{align*}
    where $B_1=\{1,\dots,n\}, B_2=\{n+1,\dots,2n\},\dots,B_n=\{(n-1)n+1,\dots,n^2\}$. So, every $|B_i|=n$, and $|U|=n^2$. Later, we will have one variable $x_u$ in the DNF $f$ for every element $u \in U$.

    Now, for every unordered pair $\{i,j\}$, let 
    \begin{align*}
        \lambda^{\{i,j\}}_i: B_i \to [n], \qquad \lambda^{\{i,j\}}_j: B_j \to [n],
    \end{align*}    
    be bijections that will suitably be chosen later in order to ensure DNF term-width $O(n)$. %
    Fix $k=\lfloor n/2 \rfloor$, and for every $S \subseteq B_i$, $|S|=k$, and $j \neq i$, define
    \begin{align*}
        m_{i,j}(S) = \min_{x \in S} \lambda^{\{i,j\}}_i(x).
    \end{align*}
    Namely, $m_{i,j}(S)$ is the smallest value of $\lambda^{\{i,j\}}_i$ over $S$. Next, we define the ``prefix set'' $A_{i, j}(S)$ to comprise of all the elements in $B_j$ that $\lambda^{\{i,j\}}_j$ maps to a value that is at most $m_{i,j}(S)$, that is
    \begin{align*}
        A_{i,j}(S)=\{y \in B_j: \lambda^{\{i,j\}}_j(y) \le m_{i,j}(S)\}.
    \end{align*}

    Now, for every $i \in [n]$ and $S \subseteq B_i, |S|=k$, define the ``positive set'' $P_{i,S}$ and ``negative set'' $N_{i,S}$ as
    \begin{align*}
        P_{i,S} := S, \qquad N_{i,S} = \left(B_i \setminus S\right) \cup \bigcup_{j \neq i} A_{i,j}(S).
    \end{align*}
    Note that $P_{i,S} \cap N_{i,S}=\emptyset$, since each $A_{i,j}(S) \subseteq B_j$, and $B_j$ is disjoint from $B_i$.
    
    The desired signed set-pair system is then $\mcT=\{(P_{T}, N_T)\}_T$, where $T=(i,S)$ ranges over all $i \in [n]$ and $S \subseteq B_i, |S|=k$. %

    \paragraph{Pairwise incompatibility.} We now argue that for every distinct pair $(i,S)$ and $(j,S')$ (where $S \subseteq B_i$ and $S' \subseteq B_j$), it holds that
    \begin{align*}
        P_{i,S} \cap N_{j,S'} \neq \emptyset \quad \text{ or } \quad P_{j,S'} \cap N_{i,S} \neq \emptyset.
    \end{align*}
    To see this, first suppose $i=j$, so that $S \neq S'$. Then, observe that $S'$ has at least one element in $B_i \setminus S$, and hence $P_{j,S'} \cap N_{i,S} \neq \emptyset$.

    Now suppose that $i \neq j$, and let
    \begin{align*}
        a = m_{i,j}(S), \qquad b=m_{j,i}(S').
    \end{align*}
    If $a \le b$, then the element $x \in S$ for which $\lambda^{\{i,j\}}_i(x) = a$ belongs to $A_{j,i}(S')$; hence
    \begin{align*}
        P_{i,S} \cap N_{j,S'} \neq \emptyset.
    \end{align*}
    On the other hand, if $a > b$, then the element $x \in S'$ for which $\lambda^{\{i,j\}}_j(x)=b$ belongs to $A_{i,j}(S)$; hence
    \begin{align*}
        P_{j,S'} \cap N_{i,S} \neq \emptyset.
    \end{align*}
    Thus, we have shown that for any distinct $(i,S)$ and $(j,S')$, it holds that at least one of $P_{i,S} \cap N_{j,S'}$  and $P_{j,S'} \cap N_{i,S}$ is non-empty.

    \paragraph{Choosing the bijections.}
    For every unordered pair $\{i,j\}$, choose both the bijections $\lambda^{\{i,j\}}_i$ and $\lambda^{\{i,j\}}_j$ independently and uniformly at random. We will argue via the probabilistic method that there exists a choice of bijections satisfying that $|P_{i,S}|+|N_{i,S}| = O(n)$ for every $(i,S)$.

    To that end, fix any $i \in [n]$, and $S \subseteq B_i, |S|=k$; note that $|P_{i,S}|=|S|=k$. We will now bound the probability of the bad event that $|N_{i,S}|$ is large. For any $j \neq i$, note that the set $\{\lambda^{\{i,j\}}_i(x):x \in S\}$ is a uniformly random subset of $[n]$ of size $k$, and the event $m_{i,j}(S) > t$ occurs only if all the members of this subset are larger than $t$; that is,
    \begin{align*}
        \Pr[m_{i,j}(S) > t] = \frac{\binom{n-t}{k}}{\binom{n}{k}} \le \left(1-\frac{t}{n}\right)^{k} \le e^{-tk/n} \le e^{-t/3}.
    \end{align*}
    Thus, we can choose $\lambda > 0$ small enough so that
    \begin{align*}
        \E\left[e^{\lambda m_{i,j}(S)}\right] \le 2.
    \end{align*}
    Now note that since $\lambda^{\{i,j\}}_i$ for every $j \neq i$ is sampled independently, the different random variables $m_{i,j}(S)$ are independent. Therefore, by Markov's inequality, we have that
    \begin{align*}
        \Pr\left[\sum_{j \neq i}m_{i,j}(S) > Cn\right] \le 2^{n-1} e^{-\lambda C n}.
    \end{align*}
    The total number of $(i,S)$ pairs is $n\binom{n}{k}\le n2^{n}$. Thus, choosing $C$ to be an appropriately large constant, we can union bound the bad event above over all these pairs, and obtain that with positive probability, it holds that $\sum_{j \neq i} m_{i,j}(S) \le Cn$ for every $(i,S)$ pair. In particular, this implies that there exists a good realization of all the bijections for which $\sum_{j \neq i} m_{i,j}(S) \le Cn$ for every $(i,S)$ pair --- we will choose this realization of the bijections for our signed set system.
    
    Next, observe that the size of any $A_{i,j}(S)$ is exactly $m_{i,j}(S)$, and hence, for every $(i,S)$,
    \begin{align*}
        \left|\bigcup_{j \neq i}A_{i,j}(S)\right| = \sum_{j \neq i} m_{i,j}(S) \le Cn.
    \end{align*}
    This means that
    \begin{align*}
        |P_{i,S}| + |N_{i,S}| = k + n-k + \left|\bigcup_{j \neq i}A_{i,j}(S)\right| \le n+Cn = O(n),
    \end{align*}
    as desired.

    \paragraph{Large hitting number for the positive sets.} Consider the family of positive sets $\mcP=\{P_{i,S}: i \in [n], S \subseteq B_i, |S|=k\}$. The hitting/transversal number $\tau(\mcP)$ of this family is the size of the smallest subset of $U$ that has non-empty intersection with every member of $\mcP$. Let $H$ be such a subset that realizes the hitting number. For any $i \in [n]$, since $H$ must intersect every $P_{i,S}=S\subseteq B_i, |S|=k$, it must be the case that
    \begin{align*}
        |H \cap B_i| \ge n-k+1 \ge n/2+1.
    \end{align*}
    Since the above holds for every $i \in [n]$, and all the $B_i$'s are disjoint, we obtain that
    \begin{align*}
        \tau(\mcP)=|H| \ge n(n/2+1) = \Omega(n^2). 
    \end{align*}

    \paragraph{Constructing the DNFs.} We now define the unambiguous DNF $f$ with respect to the signed set-pair system constructed above. Concretely, we will have a variable $x_u$ for every $u \in U$, so that $f$ maps $\{0,1\}^{n^2} \to \{0,1\}$. Let $x=(x_u)_{u \in U}$; we define
    \begin{align}
        \label{eqn:def-dnf}
        f(x) = \bigvee_{(i,S)} \left(\bigwedge_{u \in P_{i,S}} x_u \wedge \bigwedge_{v \in N_{i,S}} \neg x_{v}\right),
    \end{align}
    where the OR ranges over all $i \in [n], S \subseteq B_i, |S|=k$. Since $|P_{i,S}|+|N_{i,S}| = O(n)$ for all $(i,S)$ as argued above, the term-width of $f$ is $O(n)$. 
    
    We now argue that $f$ is an unambiguous DNF, which will imply that $\UC_1(f) = O(n)$. Towards this, fix any $x$ for which $f(x)=1$; this must mean that at least one of the conjunctions is true. Namely, there exists $(i,S)$ for which
    \begin{align*}
        \bigwedge_{u \in P_{i,S}} x_u \wedge \bigwedge_{v \in N_{i,S}} \neg x_{v} = 1,
    \end{align*}
    which in turn fixes $x_u=1$ for every $u \in P_{i,S}$, and $x_u=0$ for every $v \in N_{i,S}$. Now consider any pair $(j,S')$ that is distinct from $(i,S)$. As argued above, $(j,S')$ and $(i,S)$ are pairwise incompatible, meaning that $P_{i,S} \cap N_{j,S'} \neq \emptyset$ or $P_{j,S'} \cap N_{i,S} \neq \emptyset$. But since the variables in $P_{i,S}, N_{i,S}$ are pinned down, we have that $x_v = 1$ for some $v \in N_{j,S'}$ or $x_u=0$ for some $u \in P_{j,S'}$. In particular, this means that the conjunction corresponding to the pair $(j, S')$ evaluates to false, i.e.,
    \begin{align*}
        \bigwedge_{u \in P_{j,S'}} x_u \wedge \bigwedge_{v \in N_{j,S'}} \neg x_{v} = 0.
    \end{align*}
    Thus, exactly one term in the DNF evaluates to true for any $x$ satisfying $f(x)=1$, meaning that $f$ is an unambiguous DNF.

    Finally, we argue that $\C_0(f)=\Omega(n^2)$. Towards this, consider the input $x=0^{n^2}$, for which $f(x)=0$. Observe that any 0-certificate $\rho_x$ for $x$ must necessarily set $\rho_{x,u}=0$ for at least one $u \in P_{i,S}$ for every pair $(i,S)$. Otherwise, if $\rho_{x,u}=*$ for every $u \in P_{i,S}$ for some pair $(i,S)$, we could set all of these $*$ values to 1 and set any remaining $*$ values in $N_{i,S}$ to $0$, thus obtaining an $x'$ that is consistent with $\rho_x$ but makes the conjunction corresponding to $(i,S)$ true. This means that any 0-certificate $\rho_x$ for $x$ must have at least one non-$*$ (set to $0$) within every $P_{i,S}$. In particular, the set $\{u \in U: \rho_{x,u} \neq *\}$ has non-empty intersection with every member of the family of positive sets $\mcP$, meaning that its size is at least the hitting number $\tau(\mcP)$. We thus have that
    \begin{align*}
        \C_0(f) \ge \C_0(f, x) \ge \tau(\mcP)=\Omega(n^2).
    \end{align*}
    This completes the proof.
\end{proof}

\section{Constant-gadget Lifting}
\label{sec:lifting}

We now show how we can lift the separation between certificate complexities above to a communication problem. Crucially, we accomplish this lifting with a constant-sized gadget.

For a gadget function $g:\{0,1\}^k \times \{0,1\}^k \to \{0,1\}$ and a Boolean function $f:\{0,1\}^n \to \{0,1\}$, let $f \circ g^n: \{0,1\}^{kn} \times \{0,1\}^{kn} \to \{0,1\}$ denote the communication problem that arises from lifting $f$ with the gadget $g$, defined as
\begin{align*}
    f \circ g^n((x_1,\dots,x_n), (y_1,\dots,y_n)) = f(g(x_1,y_1),\dots,g(x_n, y_n)),
\end{align*}
where each $x_i, y_i \in \{0,1\}^k$.

We will construct a constant-sized gadget $g$ such that upon lifting our unambiguous DNFs from \Cref{thm:unambigous-dnf} through the gadget $g$, one obtains a communication problem $h$ that has large separation between two measures: $\Par_1(h)$ and $\Cov_0(h)$. Recall that $\Cov_0(h)$ is the minimum number of 0-monochromatic rectangles required to cover all the 0-entries in the communication matrix of $h$. The quantity $\Par_1(h)$ is defined to be the minimum number of 1-monochromatic rectangles required to \textit{partition} all the 1-entries in the communication matrix of $h$.

We now restate and prove \Cref{thm:constant-gadget-lifting}.

\Lifting*

Since for any $f:\{0,1\}^n \to \{0,1\}$ and any gadget $g:\{0,1\}^k \times \{0,1\}^k \to \{0,1\}$, it is known that $\log \Par_1(f \circ g^n) \le O(k \cdot \UC_1(f))$, we obtain the following corollary by combining \Cref{thm:constant-gadget-lifting} and \Cref{thm:unambigous-dnf}.

\begin{corollary}[Communication Complexity Separation]
	\label{corollary:communication-separation}
	For all large enough $n$, there exists a communication problem $h:\{0,1\}^{3n^2} \times \{0,1\}^{3n^2} \to \{0,1\}$ for which
	\begin{align*}
		\log \Par_1(h) = O(n), \qquad \log \Cov_0(h) = \Omega(n^2).
	\end{align*}
\end{corollary}

Since the rank (over the reals) of the communication matrix of $h$ is at most $\Par_1(h)$, and since every deterministic communication protocol for $h$ having cost $D(h)$ induces a cover of all its $0$-entries by at most $2^{D(h)}$ rectangles, \Cref{corollary:communication-separation} also helps us shave all the log factors in the quadratic lower bound of \citet[Theorem 1.2]{goos2018deterministic} for the log-rank conjecture, which, to the best of our knowledge, was the previous best lower bound for the conjecture.

\begin{corollary}[Quadratic Lower Bound for Log-Rank]
	\label{corollary:log-rank}
	There exists a family of communication problems $h$ for which
    \begin{align*}
        D(h) \ge \Omega\left((\log \rank(h))^2\right).
    \end{align*}
\end{corollary}

In the rest of this section, we will systematically establish \Cref{thm:constant-gadget-lifting}. We start by defining the gadget function $g$.

\subsection{Cyclic Gadget}
\label{sec:gadget-definition}

Fix $Q=2^{3}$, and let $\Z_Q=\{0,1,\dots,Q-1\}$ be the additive group modulo $Q$. We identify every element of $\Z_Q$ with its binary representation in $\{0,1\}^{3}$. The gadget function $g:\{0,1\}^{3} \times \{0,1\}^{3} \to \{0,1\}$ is then defined as
\begin{align}
	\label{eqn:gadget-def}
	g(x,y) = \begin{cases}
		0 & \text{if $y-x \in \{-1,0,1\}\mod Q$}, \\
		1 & \text{otherwise.}
	\end{cases}
\end{align}
That is, $g$ outputs $0$ only when $y$ is either $x$, or equal to one of the two neighbors of $x$ around the cycle $0,1,\dots,Q-1,0$. Note that the gadget is symmetric, namely $g(x,y)=g(y,x)$.

We recall the notation from \Cref{sec:unambigous-dnf-separation}, where the universe $U=[n^2]=B_1 \sqcup \dots \sqcup B_n$, and $B_i = \{(i-1)n+1,\dots, i n\}$. The lifted communication problem $h:\Z_Q^U \times \Z_Q^U \to \{0,1\}$ is then defined as
\begin{align*}
	h(x,y) = f((g(x_u,y_u))_{u \in U}).
\end{align*}
Since $g(x,x)=0$ for every $x$, and $f(0^{n^2})=0$, we have that $h(x,x)=0$ for all $x \in \Z_Q^U$.

\subsection{Characterizing Matrix}
\label{sec:matrix}

We will now define a matrix $M$ that is closely related to the communication matrix of $h$. The purpose of defining this matrix is that the spectrum of this matrix will govern $\log \Cov_0(h)$. Namely, as we shall see in \Cref{sec:eig-bound-to-communication} ahead, in order to lower-bound $\log \Cov_0(h)$, it will suffice to lower-bound the minimum eigenvalue of $M$.

We will express $M$ via the Kronecker product of $Q \times Q$ matrices that will serve as building blocks for $M$, and whose rows and columns we will index by $\Z_Q$. Let $I_Q$ and $J_Q$ be the $Q \times Q$ identity matrix and all-ones matrix respectively. Let $C_Q$ be the adjacency matrix of the cycle over $\Z_Q$, i.e.,
\begin{align*}
	C_Q(x,y) = 1 \iff x-y \in \{-1,1\} \mod Q
\end{align*}
Now define the two matrices $W_0$ and $W_1$ as
\begin{align*}
	&W_0 = I_Q + \frac{1}{4}C_Q = \begin{bmatrix}
		1 & 1/4 & 0 & \dots & 0 & 0 & 1/4 \\
		1/4 & 1 & 1/4& \dots & 0 & 0 & 0 \\
		\vdots \\
		0 & 0 & 0 & \dots & 1/4 & 1 & 1/4 \\
		1/4 & 0 & 0 & \dots &0 & 1/4 & 1
	\end{bmatrix} \\
	&W_1 = J_Q- I_Q-C_Q =  \begin{bmatrix}
		0 & 0 & 1 & \dots & 1 & 1 & 0 \\
		0 & 0 & 0& \dots & 1 & 1 & 1 \\
		\vdots \\
		1 & 1 & 1 & \dots & 0 & 0 & 0 \\
		0 & 1 & 1 & \dots &1 & 0 & 0
	\end{bmatrix}
\end{align*}
Both $W_0$ and $W_1$ are symmetric and non-negative.  The sum of every row in $W_0$ is equal to
\begin{align}
	\label{eqn:c_0}
	c_0 := 1 + \frac14 + \frac14 = \frac32,
\end{align}
and the sum of every row in $W_1$ is equal to
\begin{align}
	\label{eqn:w_0}
	w_0 := Q-3.
\end{align}
(As will become clearer later, we use %
the subscript 0 to denote both these quantities because they correspond to the eigenvalues of the all-ones vector, which is the eigenvector at frequency $0$ in the Fourier spectrum.)

Most importantly,
\begin{align}
	\label{eqn:W-entry-characterization}
	&W_0(x,y) > 0 \iff g(x,y)=0, \nonumber \\
	&W_1(x,y) > 0 \iff g(x,y)=1.
\end{align}
The $1/4$ scaling on $C_Q$ in the definition of $W_0$ is deliberate, and will later ensure that all the eigenvalues of $W_0$ corresponding to the Fourier spectrum are positive.

Now, for any $i \in [n]$ and for any $S \subseteq B_i, |S|=k=\lfloor n/2\rfloor$, let $b^{(i,S)} \in \{0,1\}^U$ be defined as
\begin{align*}
	b^{(i,S)} = \bone_{S}, \quad \text{i.e., } b^{(i,S)}_u = 1 \iff u \in S.
\end{align*}
Then, let $N = Q^{|U|}$, and define the $N \times N$ matrix $M_{i,S}$ as
\begin{align}
    \label{eqn:def-M-i-S}
	M_{i,S} = \bigotimes_{u \in U} W_{b^{(i,S)}_{u}}.
\end{align}
In particular, for $x=(x_u)_{u \in U}$ and $y=(y_u)_{u \in U}$, where $x_u, y_u \in \Z_Q$, the entry $M_{i,S}(x,y)$ is precisely
\begin{align}
	\label{eqn:entry-in-kronecker}
	M_{i,S}(x,y) &= \prod_{u \in U} W_{b^{(i,S)}_u}(x_u, y_u).
\end{align}
Note that every entry in $M_{i,S}$ is non-negative. Exactly $k=|S|$ terms in the product above correspond to $W_1$ matrices. Since $k \ge 1$, and the diagonal of $W_1$ is 0, we get that the diagonal of $M_{i,S}$ is 0. Also, since the row-sum of the Kronecker product is equal to the product of the row-sums of the individual matrices, we get that every row in every $M_{i,S}$ sums to $c_0^{|U|-k}w_0^k$.

Finally, define $M$ as
\begin{align*}
	M = \sum_{i=1}^n \sum_{S \in \binom{B_i}{k}} M_{i,S}.
\end{align*}
By the properties of each $M_{i,S}$, $M$ is non-negative, and has a zero diagonal. Every row in $M$ sums to
\begin{align}
	\label{eqn:d}
	d := n\binom{n}{k}c_0^{|U|-k}w_0^k = n\binom{n}{k}c_0^{|U|}\left(\frac{w_0}{c_0}\right)^k.
\end{align}

Furthermore, if $M(x,y) > 0$, then it must be the case that for some $i$ and some $S \in \binom{B_i}{k}$, $M_{i,S}(x,y) > 0$, meaning that every term in the product in \eqref{eqn:entry-in-kronecker} is positive. But this means, from \eqref{eqn:W-entry-characterization}, that for every $u \in U$,
\begin{align*}
	g(x_u, y_u) = b^{(i,S)}_u,
\end{align*}
and hence,
\begin{align*}
	h(x,y) = f((g(x_u,y_u))_{u \in U}) = f(b^{(i,S)})= f(\Ind_S) = 1.
\end{align*}
In the above, we used the precise form \eqref{eqn:def-dnf} of the DNF $f$ from \Cref{thm:unambigous-dnf}, where the input $\Ind_S$ satisfies exactly the term corresponding to $(i,S)$. We have thus established a crucial property of the matrix $M$, namely
\begin{align}
	\label{eqn:crucial-property}
	M(x,y) > 0 \implies h(x,y) = 1.
\end{align}
If we interpret $M$ as the adjacency matrix of a graph, where $x$ and $y$ are connected by an edge iff $M(x,y)>0$, then the above relation implies that a 0-entry in the communication matrix of $h$ translates to the \textit{absence} of an edge in $M$. This viewpoint will be helpful ahead.

In summary, $M$ is a real, symmetric, non-negative matrix. Its diagonal entries are zero, every row sums to $d$, and every positive entry in $M$ corresponds to a 1-input for $h$.

\subsection{Bounded Spectrum to Communication Lower Bound}
\label{sec:eig-bound-to-communication}

We now show how a lower bound on the minimum eigenvalue of the matrix $M$ translates to a lower bound on $\log \Cov_0(h)$ for the lifted communication problem $h$. Towards this, as mentioned above, we will first interpret the matrix $M$ as the adjacency matrix of a graph. Namely, we will consider an undirected graph $G_M$ on $N$ vertices, where two vertices $x$ and $y$ are connected by an (unweighted) edge iff $M(x,y) > 0$. 

Recall that every row in $M$ has a common sum $d > 0$. The following lemma states that no independent set in $G_M$ can be larger than $\eps N$, if every eigenvalue of $M$ is at least $-d\eps$ for some $\eps \ge 0$.

\begin{lemma}[Bounded Independent Sets]
	\label{lemma:eig-to-independent-set}	
	Let $\lambda_{\min}$ be the minimum eigenvalue of $M$, and suppose that $\lambda_{\min} \ge -d\eps$, for some $\eps \ge 0$. Then, every independent set $S$ in the graph $G_M$ satisfies that
	\begin{align*}
		|S| \le \left(\frac{\eps}{1+\eps}\right) N \le \eps N.
	\end{align*}
\end{lemma}
\begin{proof}
	Fix any independent set $S$ in $G_M$. Let $\bone$ be the all-ones vector, and let $\bone_S$ be the indicator vector of the set $S$ (i.e., $\bone_{S,i}=1$ iff $i \in S$). Let $z$ be the component of $\bone_S$ that is orthogonal to $\bone$, and observe that
	\begin{align*}
		\bone_S = \frac{|S|}{N} \bone + z.
	\end{align*}
	For convenience, let $\delta := \frac{|S|}{N}$. We then have that
	\begin{align}
		\label{eqn:norm-z}
		\|\bone_S\|_2^2 =\delta N = \delta^2 N + \|z\|_2^2 \implies \|z\|_2^2 = \delta N (1-\delta).
	\end{align}
	Now, by the variational form for the minimum eigenvalue of $M$, it holds that
	\begin{align}
		\label{eqn:variational-bound}
		z^T M z \ge \lambda_{\min} \|z\|_2^2 \ge -d\eps\delta N (1-\delta),
	\end{align}
	where we used our assumption on the minimum eigenvalue of $M$. On the other hand, note that since $S$ is an independent set in $G_M$, and $M$ has a zero diagonal,
	\begin{align*}
		\bone_S^T M \bone_S &= \sum_{x,y \in S} M(x,y) =0.
	\end{align*}
	This gives us that
	\begin{align}
		z^TMz &= (\bone_S-\delta \bone)^T M (\bone_S-\delta \bone) \nonumber \\
		&= \delta^2Nd - 2\delta \bone_S^TM\bone = \delta^2Nd - 2\delta^2 Nd \nonumber \\
		&= -\delta^2 N d. \label{eqn:exact-variational}
	\end{align}
	Combining \eqref{eqn:variational-bound} and \eqref{eqn:exact-variational}, we get that
	\begin{align*}
		\delta^2Nd -d\eps\delta N (1-\delta) \le 0.
	\end{align*}
	Now, if $\delta = 0$, the claimed bound holds immediately. Otherwise, we can divide throughout by $\delta Nd$ to get
	\begin{align*}
		\delta -\eps(1-\delta) \le 0 \implies \delta \le \frac{\eps}{1+\eps} \le \eps.
	\end{align*}
\end{proof}

We conclude this subsection by arguing how bounded independent sets give us the desired lower bound on $\log \Cov_0(h)$. Let $R=X \times Y$ be any 0-monochromatic rectangle in the communication matrix of $h$, where $X,Y \subseteq \Z_Q^U$. Let $S_R$ be the diagonal elements in $R$, namely
\begin{align*}
	S_R = \{x \in \Z_Q^U: (x,x) \in R\} = X \cap Y.
\end{align*}
For any $x,y \in S_R$, it holds that $(x,y) \in R$ since $R$ is a rectangle, and hence $h(x,y)=0$, since $R$ is 0-monochromatic. From \eqref{eqn:crucial-property}, it must then hold that $M(x,y)=0$. We thus have that $M(x,y)=0$ for every $x,y \in S_R$, meaning that $S_R$ is an independent set in $G_M$. From \Cref{lemma:eig-to-independent-set}, this means that $|S_R| \le \eps N$.

So, let $\mcF$ %
be any collection of $0$-monochromatic rectangles that cover every $0$-entry in the communication matrix of $h$. Since $h(x,x)=0$ for all $x$, the collection must cover all the $Q^{|U|}=N$ diagonal entries of the communication matrix. Therefore, by the union bound
\begin{align*}
	N \le \sum_{R \in \mcF} |S_{R}| \le \eps N|\mcF| \implies |\mcF| \ge \frac{1}{\eps}.
\end{align*}
Since the bound above holds for every $0$-monochromatic rectangle cover, we have that
\begin{align*}
	\Cov_0(h) \ge \frac{1}{\eps} \implies \log \Cov_0 (h) \ge \log(1/\eps).
\end{align*}
Thus, if we can show that the minimum eigenvalue of $M$ is lower-bounded by $-d\eps$ for $\eps = \rho^{\gamma n^2}$ for some constants $0 < \rho < 1$ and $\gamma >0$, we will have obtained
\begin{align*}
	\log \Cov_0 (h) \ge \gamma n^2 \log(1/\rho) = \Omega(n^2),
\end{align*}
which is the desired bound of \Cref{thm:constant-gadget-lifting}. Therefore, the remaining pursuit is to show that the minimum eigenvalue of $M$ is lower-bounded by a value that is not too negative and pretty close to 0, namely $-d\rho^{\gamma n^2}$. 

\subsection{Fourier Spectrum}
\label{eqn:fourier-spectrum}

We will now compute the eigenvalues of the matrix $M$. For this, we will first precisely compute the eigenvalues of the matrices $W_0$ and $W_1$, corresponding to eigenvectors in the \textit{Fourier spectrum}.

Concretely, for $r \in \Z_Q$, let $\chi_r \in \mathbb{C}^{Q}$ denote the vector at frequency $r$ in the Fourier spectrum. The size of $\chi_r$ is $Q$, and for $x \in \Z_Q$, the $x^{\text{th}}$ entry of $\chi_r$ is defined as
\begin{align}
	\chi_r(x) = \omega^{rx} = \exp\left(\frac{2\pi i rx}{Q}\right),
\end{align}
where $\omega=\exp\left(\frac{2\pi i }{Q}\right)$ is the primitive $Q^\text{th}$ root of unity, and satisfies that $\omega^Q=1$, $\omega^r \neq 1$ for $1 \le r < Q$, and $\sum_{r=0}^{Q-1}\omega^r=0$. For $r,s \in \Z_Q$, note that $\langle \chi_r, \chi_s\rangle=\sum_{x \in \Z_Q} \chi_r(x) \overline{\chi_s(x)}$ is equal to $Q$ if $r=s$, and $0$ if $r\neq s$. Thus, the $Q$-many $\chi_r$ vectors are mutually orthogonal and non-zero. We next argue that they form an eigenbasis for $W_0$ and $W_1$.

Recall that $W_0$ and $W_1$ are defined in terms of the identity matrix $I_Q$, the all-ones matrix $J_Q$ and the cycle matrix $C_Q$; each of these matrices is a \textit{circulant} matrix. Namely, the entry at position $(x,y)$ depends only on $y-x$ (modulo $Q$). For any circulant matrix $A$, where $A(x,y)= a(y-x)$, every $\chi_r$ is an eigenvector, since
\begin{align*}
	(A\chi_r)(x) &= \sum_{y \in \Z_Q} a(y-x)\chi_r(y) \\
	&= \sum_{t \in \Z_Q} a(t)\chi_r(x+t) = \sum_{t \in \Z_Q} a(t)\chi_r(x)\chi_r(t) \\
	&= \chi_r(x) \sum_{t \in \Z_Q}a(t)\chi_r(t).
\end{align*}
Thus, every $\chi_r$ is an eigenvector for $W_0$ as well as $W_1$.\footnote{A similar calculation also appears in the recent work of \citet{larsen2026sample} which uses spectral analysis over Cayley graphs to establish lower bounds in replicable learning.} 

For the identity matrix $I_Q$, each $\chi_r$ has eigenvalue $1$. For the all-ones matrix $J_Q$,  each $\chi_r$ has eigenvalue equal to $\sum_{x \in \Z_Q}\chi_r(x)$, which is equal to $Q$ if $r=0$, and $0$ if $r \neq 0$. Finally, for the cycle matrix $C_Q$, 
\begin{align*}
	(C_Q\chi_r)(x) = \chi_r(x-1) + \chi_r(x+1)= (\omega^r + \omega^{-r}) \cdot \chi_r(x) = 2\cos(2\pi r/Q) \chi_r(x).
\end{align*}

In total, we obtain that the eigenvalues $c_r$ and $w_r$ for the matrices $W_0$ and $W_1$, corresponding to the eigenvector $\chi_r$, are equal to
\begin{align}
	&c_r = 1 + \frac12\cos(2\pi r/Q), \label{eqn:c_r}\\
	&w_r = \begin{cases}
		Q - 3 & \text{if $r=0$}, \\
		-1 - 2\cos(2\pi r/Q) & \text{if $r \neq 0$}. \label{eqn:w_r}
	\end{cases}
\end{align}
As also alluded to earlier, $\chi_0$ is the all-ones vector, and hence has eigenvalue equal to the row-sums (namely $c_0=3/2$ and $w_0=Q-3$ from \eqref{eqn:c_0} and \eqref{eqn:w_0}) for $W_0$ and $W_1$ respectively.

Note that every $c_r$ is positive; in particular, $\frac12 \le c_r \le \frac32$. The $1/4$ scaling on the entries in $C_Q$ was useful precisely for this.

Now, for a frequency vector $\xi=(\xi_u)_{u \in U} \in \Z_Q^{|U|}$, where every $\xi_u \in \Z_Q$, consider the tensor product vector $\chi_\xi = \bigotimes_{u \in U} \chi_{\xi_u} \in \mathbb{C}^{N}$, where for $x=(x_u)_{u \in U}$, $x_u \in \Z_Q$, the $x^\text{th}$ entry of $\chi_\xi$ is defined as
\begin{align}
    \label{eqn:tensor-eigenvectors}
    \chi_\xi(x) = \prod_{u \in U} \chi_{\xi_u}(x_u) = \exp\left(\frac{2\pi i}{Q}\sum_{u \in U}\xi_u x_u\right).
\end{align}
Since $\{\chi_r\}_{r \in \Z_Q}$ form a basis for $\mathbb{C}^{Q}$, the $Q^{|U|}=N$ many vectors $\{\chi_\xi\}_{\xi \in \Z_Q^{|U|}}$ form a basis for $\mathbb{C}^{N}$. We will now argue that these are eigenvectors for the matrix $M$. Indeed, for any tensor product matrix $A = \bigotimes_{u \in U} A_u$, where for each $A_u$, $\chi_{\xi_u}$ is an eigenvector with eigenvalue $\mu_u$, it holds that
\begin{align*}
    (A\chi_\xi)(x) &= \sum_{y \in \Z_Q^{|U|}} \prod_{u \in U} A_u(x_u, y_u)\chi_{\xi_u}(y_u) \\
    &= \prod_{u \in U} \left( \sum_{y_u \in \Z_Q} A_u(x_u, y_u)\chi_{\xi_u}(y_u)\right) \\
    &= \prod_{u \in U} \mu_u \chi_{\xi_u}(x_u) \\
    &= \left(\prod_{u \in U} \mu_u\right) \chi_\xi(x).
\end{align*}
Thinking of $A$ as $M_{i,S}$, where, recalling \eqref{eqn:def-M-i-S}, each $A_u$ is either $W_0$ or $W_1$, and $\chi_{\xi_u}$ is an eigenvector for $A_u$ with eigenvalue either $c_{\xi_u}$ or $w_{\xi_u}$, we get from the above that any $\chi_\xi$ is an eigenvector for $M_{i,S}$ with eigenvalue
\begin{align*}
    \prod_{u \in S}w_{\xi_u} \prod_{u \notin S} c_{\xi_u}.
\end{align*}
Since $M$ is the sum of $M_{i,S}$ over all $(i,S)$, we get that $\chi_\xi$ is an eigenvector for $M$ with eigenvalue $\lambda_\xi$, where
\begin{align}
    \lambda_\xi &= \sum_{i=1}^n \sum_{S \in \binom{B_i}{k}}\prod_{u \in S}w_{\xi_u} \prod_{u \notin S} c_{\xi_u} 
    = \left(\prod_{u \in U} c_{\xi_u}\right) \sum_{i=1}^n \sum_{S \in \binom{B_i}{k}}\prod_{u \in S} \frac{w_{\xi_u}}{c_{\xi_u}}. \label{eqn:M-eigenvalue-expression}
\end{align}
In the above, dividing by $c_{\xi_u}$ is safe since every $c_{\xi_u} > 0$.

The product term inside the double summation above is an elementary symmetric polynomial of degree $k$. Namely, let
\begin{align*}
    e_k(y_1,\dots,y_t) = \sum_{T \in \binom{[t]}{k}} \prod_{j \in T} y_j
\end{align*}
denote the elementary symmetric polynomial of degree $k$ in the variables $y_1,\dots,y_t$. Note that $e_k(\alpha y_1,\dots,\alpha y_t)=\alpha^k e_k(y_1,\dots,y_t)$. Thus, if we define $q_{r}=\frac{w_r/w_0}{c_r/c_0}$, we have that
\begin{align*}
    \lambda_\xi &= \left(\prod_{u \in U} c_{\xi_u}\right) \sum_{i=1}^n e_k\left(\left(\frac{w_{\xi_u}}{c_{\xi_u}}\right)_{u \in B_i}\right) \\
    &= \left(\prod_{u \in U} c_{\xi_u}\right) \sum_{i=1}^n e_k\left(\left(q_{\xi_u} \cdot \frac{w_0}{c_0}\right)_{u \in B_i}\right) \\
    &= \left(\prod_{u \in U} c_{\xi_u}\right) \left( \frac{w_0}{c_0}\right)^k \sum_{i=1}^n e_k\left(\left(q_{\xi_u}\right)_{u \in B_i}\right).
\end{align*}
Recalling the expression for the common row-sum $d$ of $M$ from \eqref{eqn:d}, we obtain the following expression for the \textit{normalized} eigenvalue $\frac{\lambda_\xi}{d}$:
\begin{align}
    \frac{\lambda_\xi}{d} = \left(\prod_{u \in U} \frac{c_{\xi_u}}{c_0}\right)\frac{1}{n}\sum_{i=1}^n \underbrace{\frac{e_k\left(\left(q_{\xi_u}\right)_{u \in B_i}\right)}{\binom{n}{k}}}_{P_i} = \left(\prod_{u \in U} \frac{c_{\xi_u}}{c_0}\right)\frac{1}{n}\sum_{i=1}^n P_i, \label{eqn:normalized-eigenvalue}
\end{align}
where we defined
\begin{align}
    \label{eqn:def-P-i}
    P_i := \frac{e_k\left(\left(q_{\xi_u}\right)_{u \in B_i}\right)}{\binom{n}{k}}.
\end{align}
Next, let $\ell_i$ be the number of non-zero $\xi_u$ values within $B_i$, and $L$ be the number of non-zero $\xi_u$ values overall. Namely,
\begin{align}
    \label{eqn:def-ell-L}
    \ell_i := \left|\left\{u \in B_i: \xi_u \neq 0\right\}\right|, \quad L := \left|\left\{u \in U: \xi_u \neq 0\right\}\right| = \sum_{i=1}^n \ell_i.
\end{align}
Since $q_0=1$, $P_i$ may be written as
\begin{align}
    \label{eqn:P-i-in-terms-of-ell-i}
    P_i = \frac{e_k\left(1^{n-\ell_i}, (q_{\xi_u})_{u \in B_i: \xi_u \neq 0} \right)}{\binom{n}{k}}.
\end{align}
In the above $1^{n-\ell_i}=\underbrace{(1,\dots,1)}_{\text{$n-\ell_i$ times}}$. We next bound the values of $q_{\xi_u}$ where $\xi_u \neq 0$. For any $r \in \Z_Q, r \neq 0$, and $x=\cos(2\pi r/Q) \in [-1,1]$, we have that
\begin{align*}
    |q_r| = \left|\frac{w_r/w_0}{c_r/c_0}\right| &= \frac{(3/2)|1+2x|}{(Q-3)(1+x/2)}. \tag{Using \eqref{eqn:c_r},\eqref{eqn:w_r}}
\end{align*}
The denominator above is always positive. Now, if $x \ge -1/2$, we have that 
\begin{align*}
\frac{(3/2)|1+2x|}{1+x/2}=\frac{(3/2)(1+2x)}{1+x/2}=3\left(2-\frac{3}{2+x}\right)\le 3. \tag{since $x \le 1$}
\end{align*}
On the other hand, if $x < -1/2$, we have that
\begin{align*}
    \frac{(3/2)|1+2x|}{1+x/2}=\frac{-(3/2)(1+2x)}{1+x/2} = 3\left(\frac{3}{2+x}-2\right)\le 3. \tag{since $x \ge -1$}
\end{align*}
Thus, it always holds that $\frac{(3/2)|1+2x|}{1+x/2} \le 3$, which gives us that
\begin{align}
    |q_r| \le \theta := \frac{3}{Q-3} = \frac{3}{5}. %
\end{align}
Namely, the $q_{\xi_u}$ values where $\xi_u \neq 0$ are at most $\theta=3/5$ in absolute value. It thus remains to reason about polynomials of the form $e_k(1^{n-\ell_i}, t_1,\dots,t_{\ell_i})$, where every $|t_i| \le \theta$. This is the focus of the next subsection.

\subsection{Bounding the Polynomial}
\label{sec:polynomial-analysis}

Before proceeding to the exact technical analysis, we provide some intuition on what kinds of polynomial bounds we'd like to derive. Recall first the expression for the normalized eigenvalue $\lambda_\xi/d$ from \eqref{eqn:normalized-eigenvalue}:
\begin{align*}
    \frac{\lambda_\xi}{d} = \left(\prod_{u \in U} \frac{c_{\xi_u}}{c_0}\right)\frac{1}{n}\sum_{i=1}^n P_i.
\end{align*}
Assume for the moment that $L$ --- the number of non-zero $\xi_u$ values --- is positive. Then, observe that
\begin{align*}
    \left(\prod_{u \in U} \frac{c_{\xi_u}}{c_0}\right) = \left(\prod_{u\,:\,\xi_u \neq 0} \frac{c_{\xi_u}}{c_0}\right) \le \left(\max_{r \in \Z_Q \,:\, r \neq 0} \frac{c_r}{c_0}\right)^L \le \rho^L
\end{align*}
for some $\rho < 1$, since for $r \neq 0$, we have that $c_r/c_0 < 1$ (see \eqref{eqn:c_r}). Now, if $\lambda_\xi < 0$, it must be the case that $\frac{1}{n}\sum_{i=1}^n P_i < 0$. Since the smallest that this fraction can be is $-1$, we get that
\begin{align*}
     \frac{\lambda_\xi}{d} \ge -\prod_{u \in U} \frac{c_{\xi_u}}{c_0} \ge -\rho^L.
\end{align*}
So, if we can show that $L=\sum_i \ell_i \ge \Omega(n^2)$ whenever $\sum_{i=1}^n P_i < 0$, we will have obtained our desired lower bound on $\lambda_\xi$ (if $\lambda_\xi$ is positive, the lower bound holds vacuously).

If the number of blocks $i$ for which $\ell_i > \Omega(n)$ (``bad'' blocks) is at least $n/2$, we immediately have that $L \ge \Omega(n^2)$. Otherwise, the number of blocks for which $\ell_i \le \Omega(n)$ (``good'' blocks) is larger than $n/2$. Here, we note that, since each of the $\ell_i$ entries in $P_i$ that are \textit{not} equal to 1 are at most $\theta$ in absolute value, we can show that $|P_i| \le 2^{-\Theta(\ell_i)}$. In particular, for the bad blocks, since $\ell_i > \Omega(n)$, this means that $|P_i| \le 2^{-\Theta(n)}$. Suppose we could now show that any good block satisfies $P_i \ge 2^{-\Theta(\ell_i)} > 0$. Then, if $\sum_i P_i$ is negative, it \textit{must} be the case that the contribution of bad blocks (in magnitude) is strictly larger than the contribution of the good blocks. Combining this with the upper bound on $|P_i|$ for bad blocks and the lower bound on $P_i$ for good blocks then lets us conclude that $L \ge \Omega(n^2)$.

This sketch leads us to proving the following lemma:

\begin{lemma}[Symmetric Polynomial Bound]
    \label{lemma:elementary-polynomial-bound}
    Fix $n \ge 88$, any $\ell$ satisfying $0 \le \ell \le n$, any $t_1,\dots,t_{\ell} \in [-\theta, \theta]$ where $\theta = \frac{3}{5}$, and any $k$ satisfying $0 \le k \le n/2$. Let $P=\frac{e_k(1^{n-\ell}, t_1,\dots,t_{\ell})}{\binom{n}{k}}$.\footnote{We use the convention that $e_0(y_1,\dots,y_t)=1$.} Then:
    \begin{enumerate}
        \item[(a)] If $k \ge n/3$, then $|P| \le (0.88)^\ell$. %
        \item[(b)] If $\ell \le n/128$, then $P \ge (0.12)^\ell > 0$. %
    \end{enumerate}
\end{lemma}
Part (a) of the lemma above follows from a straightforward application of Maclaurin's inequality. Part (b) is much more involved, although it only requires elementary analysis, together with some careful definitions. Before delving into the technical details though, we will sketch the high-level argument here.

The starting point is to view the elementary symmetric polynomial $e_k(y_1,\dots,y_n)$ in the context of its generating function. That is, it can be readily checked that $e_k(y_1,\dots,y_n)$ is precisely the coefficient of $z^k$ in
\begin{align*}
    \prod_{i=1}^n (1+y_i z).
\end{align*}
In particular, instantiating this with $(y_1,\dots,y_n)=(1^{n-\ell}, t_1,\dots,t_\ell)$, we get that $e_k(1^{n-\ell}, t_1,\dots,t_\ell)$ is equal to the coefficient of $z^k$ in 
\begin{align*}
    (1+z)^{n - \ell}\prod_{i=1}^\ell (1+t_i z).
\end{align*}
Let us consider the expression above at the boundary $t_1=\dots=t_\ell=-\theta$, which is equal to
\begin{align*}
    (1+z)^{n - \ell}(1-\theta z)^\ell.
\end{align*}
As it will turn out later, $e_k(1^{n-\ell}, t_1,\dots,t_\ell)$ is increasing in every $t_i$, and hence, we may focus on lower-bounding it at this boundary. Namely, we wish to lower-bound $e_k(1^{n-\ell}, (-\theta)^\ell)$, which is equal to the coefficient of $z^k$ in  $(1+z)^{n - \ell}(1-\theta z)^\ell$, by an exponentially small quantity in $\ell$ (up to the multiplicative $\binom{n}{k}$ factor). This motivates defining the polynomial
\begin{align*}
    F_r(z)=(1+z)^{n-\ell}(1-\theta z)^r = \sum_{j} C_j^{(r)}z^j,
\end{align*}
and aiming to show that $C_k^{(r+1)} \ge c \cdot C_k^{(r)}$ for some $c < 1$, since by recursively invoking this inequality $\ell$ times, we will achieve our desired lower bound on $e_k(1^{n-\ell}, (-\theta)^\ell)=C_k^{(\ell)}$. However, the polynomials $F_r$ themselves satisfy the relation
\begin{align*}
    F_{r+1}(z) = (1-\theta z)F_r(z), 
\end{align*}
which, upon comparing coefficients, gives
\begin{align*}
    C^{(r+1)}_j = C^{(r)}_j - \theta C^{(r)}_{j-1}.
\end{align*}
So, in order to show that $C_k^{(r+1)} \ge c \cdot C_k^{(r)}$, it suffices to instead show that $C^{(r)}_j \ge c' \cdot C^{(r)}_{j-1}$ for some $c' < 1$, for all $j \le k$. These steps are precisely what the following analysis concretizes.

\begin{proof}
    We first establish part (a). Suppose $k \ge n/3$: we have that
    \begin{align*}
        |P| &\le \frac{e_k(1^{n-\ell}, |t_1|,\dots,|t_{\ell}|)}{\binom{n}{k}} \tag{triangle inequality} \\
        &\le \frac{e_k(1^{n-\ell},\theta^\ell)}{\binom{n}{k}} \tag{where $\theta^\ell = \underbrace{(\theta,\dots,\theta)}_{\text{$\ell$ times}}$} \\
        &\le \left(1-(1-\theta)\frac{\ell}{n}\right)^k \tag{Maclaurin's inequality} \\
        &\le \exp\left(-(1-\theta)\frac{k\ell}{n}\right) \tag{$1+x \le e^x$ for all real $x$} \\
        &\le \exp\left(\frac{(\theta-1)\ell}{3}\right) \tag{since $k \ge n/3$} \\
        &\le (0.88)^\ell. \tag{substituting $\theta=3/5$}
    \end{align*}
    We now move on to part (b), which we will establish in stages.  

    \paragraph{Bounding a certain coefficient ratio.} Let $\ell \le n/128$, and define $s := n - \ell$. Note that $127n/128 \le s \le n$. Now, for any $0 \le r \le \ell$, define the univariate polynomial
    \begin{align}
        \label{eqn:F-r}
        F_r(z) = (1+z)^s(1-\theta z)^r = \sum_{j \ge 0}C^{(r)}_j z^j,
    \end{align}
    where $C^{(r)}_j=\sum_{p=0}^j \binom{s}{p} \binom{r}{j-p}(-\theta)^{j-p}$ is the coefficient of $z^j$ in $F_r(z)$.\footnote{We use the convention that $\binom{0}{0}=1$ and $\binom{a}{b}=0$ for $a < b$.} We note that $F_r(z)$ is a polynomial in $z$ of degree $s+r$. Then, we have the following crucial claim which lower-bounds the ratio $C^{(r)}_j/C^{(r)}_{j-1}$:
    \begin{claim}[Coefficient Ratio]
        \label{claim:coefficient-ratio}
        For any $j$ satisfying $1 \le j \le k$,
        \begin{align*}
            C^{(r)}_j \ge \frac45 C^{(r)}_{j-1} > 0.
        \end{align*}
    \end{claim}
    \begin{proof}
        We will induct on $j$, keeping $r$ fixed. As the base case, consider $j=1$. We have that $C^{(r)}_0=1$, and
        \begin{align*}
            C^{(r)}_1 &= s-r\theta = n-\ell-r\theta \\
            &\ge  n-\ell-\ell\theta = n-\frac{8\ell}{5} \ge n-\frac{8n}{640} \tag{since $\ell \le n/128$}\\
            &= \frac{79n}{80} \ge \frac{4}{5} = \frac{4}{5}C^{(r)}_0.
        \end{align*}
        Thus, the base case holds.%

        We now proceed to the inductive step. Suppose that the claim is true for $j$, where $1 \le j \le k-1$; we will then prove the claim for $j+1$, meaning
        \begin{align*}
            C^{(r)}_{j+1} \ge \frac45 C^{(r)}_{j}.
        \end{align*}
        Recall the definition \eqref{eqn:F-r} of the polynomial $F_r(z)$, defined for all real $z$:
        \begin{align*}
            F_r(z) = (1+z)^s(1-\theta z)^r = \sum_{j \ge 0}C^{(r)}_j z^j.
        \end{align*}
        Differentiating, we have that
        \begin{align*}
            F'_r(z) = -r\theta(1+z)^s(1-\theta z)^{r-1} + s(1+z)^{s-1}(1-\theta z)^r.
        \end{align*}
        Multiplying on both sides by $(1+z)(1-\theta z)$ and collecting terms gives
        \begin{align}
            \label{eqn:F-r-F'-r-equality}
            (1+z)(1-\theta z)F'_r(z) = (s-r\theta-\theta(s+r)z)F_r(z).
        \end{align}
        Since \eqref{eqn:F-r-F'-r-equality} is a polynomial identity which holds for all $z$, the coefficients of every $z^j$ must match on both sides. The coefficient of $z^j$ on the left side is
        \begin{align*}
            (j+1)C^{(r)}_{j+1} + (1-\theta)jC^{(r)}_j -\theta (j-1)C^{(r)}_{j-1}.
        \end{align*}
        The coefficient of $z^{j}$ on the right side is
        \begin{align*}
            (s-r\theta)C^{(r)}_j - \theta(s+r)C^{(r)}_{j-1}.
        \end{align*}
        Equating, we get
        \begin{align*}
            (j+1)C^{(r)}_{j+1} = (s-r\theta-(1-\theta)j)C^{(r)}_j -\theta(s+r-j+1)C^{(r)}_{j-1}.
        \end{align*}
        Note that $s+r-j+1 \ge \frac{127n}{128}-\frac{n}{2}+1 > 0$; furthermore, by the inductive hypothesis, $0 < C^{(r)}_{j-1} \le \frac54C^{(r)}_j$, giving
        \begin{align*}
            (j+1)C^{(r)}_{j+1} &\ge \left(s-r\theta-(1-\theta)j - \frac54\theta(s+r-j+1)\right)C^{(r)}_{j}.
        \end{align*}
        It then remains to show that $s-r\theta-(1-\theta)j - \frac54\theta(s+r-j+1) \ge \frac{4(j+1)}{5}$. Towards this, we have that
        \begingroup
        \allowdisplaybreaks
        \begin{align*}
            &s-r\theta-(1-\theta)j - \frac54\theta(s+r-j+1) - \frac{4(j+1)}{5} \\
            =\,&\left(1-\frac54\theta\right)s-\frac94r\theta -j\left(\frac95-\frac94\theta\right) -\left(\frac54\theta+\frac45\right) \\
            \ge\,&\left(1-\frac54\theta\right)\frac{127n}{128}-\frac{9n\theta}{512}-j\left(\frac95-\frac94\theta\right) -\left(\frac54\theta+\frac45\right) \tag{since $s \ge \frac{127n}{128}$, $r \le \ell \le \frac{n}{128}$} \\
            \ge\,&\left(1-\frac54\theta\right)\frac{127n}{128}-\frac{9n\theta}{512}-\left(\frac95-\frac94\theta\right)\left(\frac{n}{2}-1 \right) -\left(\frac54\theta+\frac45\right) \tag{since $j \le k-1 \le n/2-1$} \\
            =\,& \frac{n}{80}-\frac{11}{10} \ge 0 \tag{substituting $\theta=3/5$},
        \end{align*}
        \endgroup
        where the last inequality holds true when $n \ge 88$. Thus, the inductive step is complete, proving the claim.
    \end{proof}
    \paragraph{The case where $t_1=\dots=t_{\ell}=-\theta$.} We now establish the bound in part (b) in the special case where $t_1=\dots=t_{\ell}=-\theta$; later, we will show that this suffices for the general case.

    Note that in this case,
    \begin{align}
        \label{eqn:P-in-special-case}
        P = \frac{e_k(1^{s},(-\theta)^\ell)}{\binom{n}{k}} = \frac{\sum_{p=0}^{k} \binom{s}{p}\binom{\ell}{k-p}(-\theta)^{k-p}}{\binom{n}{k}}.
    \end{align}
    Now, for $0 \le r < \ell$, recall:
    \begin{align*}
        &F_r(z) = (1+z)^s(1-\theta z)^r = \sum_{j \ge 0}C^{(r)}_jz^j, \\
        &F_{r+1}(z) = (1+z)^s(1-\theta z)^{r+1} = \sum_{j \ge 0}C^{(r+1)}_jz^j.
    \end{align*}
    Fix any $j$, where $1 \le j \le k$. We have:
    \begin{align}
        C^{(r)}_j = \sum_{p=0}^j \binom{s}{p} \binom{r}{j-p}(-\theta)^{j-p} \label{eqn:C-r-j},
    \end{align}
    and hence,
    \begin{align*}
        C^{(r)}_j - \theta C^{(r)}_{j-1} &= \sum_{p=0}^{j-1} \binom{s}{p}(-\theta)^{j-p} \left[\binom{r}{j-p}+ \binom{r}{j-1-p}\right] + \binom{s}{j} \binom{r}{0} \\
        &= \sum_{p=0}^{j-1} \binom{s}{p} \binom{r+1}{j-p}(-\theta)^{j-p} + \binom{s}{j} \binom{r+1}{0} 
        = \sum_{p=0}^{j} \binom{s}{p} \binom{r+1}{j-p}(-\theta)^{j-p} = C^{(r+1)}_j.
    \end{align*}
    But, from \Cref{claim:coefficient-ratio}, it also holds that $0 < C^{(r)}_{j-1} \le \frac54C^{(r)}_j$; thus, we get
    \begin{align*}
        C^{(r+1)}_j \ge \left(1-\frac54\theta\right)C^{(r)}_j = \frac14 C^{(r)}_j.
    \end{align*}
    The bound above holds also for $j=0$. Iterating, we thus get that for every $j$ satisfying $0 \le j \le k$,
    \begin{align}
        \label{eqn:C-iterated-bound}
        C^{(\ell)}_j \ge \left(\frac14\right)^\ell C^{(0)}_j = \left(\frac14\right)^\ell \binom{s}{j}.
    \end{align}
    Finally, recalling the expression \eqref{eqn:P-in-special-case} for $P$, we get
    \begin{align*}
        P &=\frac{e_k(1^{s},(-\theta)^\ell)}{\binom{n}{k}} = \frac{\sum_{p=0}^{k} \binom{s}{p}\binom{\ell}{k-p}(-\theta)^{k-p}}{\binom{n}{k}} \\
        &= \frac{C^{(\ell)}_k}{\binom{n}{k}} \ge \frac{\left(\frac14\right)^\ell \binom{s}{k}}{\binom{n}{k}} \tag{using \eqref{eqn:C-r-j} and \eqref{eqn:C-iterated-bound}} \\
        &=  \left(\frac14\right)^\ell \frac{\binom{n-\ell}{k}}{\binom{n}{k}} = \left(\frac14\right)^\ell \prod_{p=0}^{\ell-1} \frac{n-k-p}{n-p} \\
        &\ge \left(\frac14\right)^\ell \prod_{p=0}^{\ell-1} \frac{\frac{n}{2}-p}{n-p} \tag{since $k \le n/2$} \\
        &\ge \left(\frac14\right)^\ell \left(\frac{63}{127}\right)^\ell \tag{$p \le \ell \le n/128 \implies \frac{\frac{n}{2}-p}{n-p} \ge \frac{63}{127}$} \\
        &\ge (0.12)^\ell.
    \end{align*}

    \paragraph{The general case where $t_1,\dots,t_{\ell}\in [-\theta, \theta]$.} We will now conclude the proof of part (b) by arguing that for any $k \ge 0$ and $\ell \ge 0$, it holds simultaneously for all $n \ge 88$ satisfying $k \le n/2$ and $\ell \le n/128$ that
    \begin{align*}
        P=\frac{e_k(1^{n-\ell}, t_1,\dots,t_{\ell})}{\binom{n}{k}} \ge (0.12)^\ell %
    \end{align*}
    whenever $t_1,\dots,t_\ell \in [-\theta,\theta]$. We will establish this by induction on $k+\ell$.
    For the base cases, consider: (1) $k=0$, whereby $P=1$ (by convention), and (2) $\ell=0$, whereby also $P=\frac{\binom{n}{k}}{\binom{n}{k}}=1$. Thus, the base cases hold. 

    So, let $k,\ell \ge 1$, and assume as our inductive hypothesis the following: for any $\ell', k' \ge 0$, where $k'+\ell' < k+\ell$, it holds simultaneously for all $n \ge 88$ satisfying $\ell' \le n/128$ and $k' \le n/2$ that
    \begin{align*}
        P' = \frac{e_{k'}(1^{n-\ell'}, t_1,\dots,t_{\ell'})}{\binom{n}{k'}} \ge (0.12)^{\ell'},
    \end{align*}
    whenever $t_1,\dots,t_{\ell'} \in [-\theta, \theta]$. We then wish to show that simultaneously for all $n \ge 88$ satisfying $\ell \le n/128$ and $k \le n/2$, it holds that 
    \begin{align*}
        P=\frac{e_k(1^{n-\ell}, t_1,\dots,t_{\ell})}{\binom{n}{k}} \ge (0.12)^\ell %
    \end{align*}
    whenever $t_1,\dots,t_\ell \in [-\theta,\theta]$.
    
    So, fix any $n \ge 88$ satisfying $\ell \le n/128$ and $k \le n/2$, %
    and consider $P=\frac{e_{k}(1^{n-\ell}, t_1,\dots,t_{\ell})}{\binom{n}{k}}$ as a polynomial in $t_1,\dots,t_\ell$. The partial derivative with respect to any $t_j$ is
    \begin{align*}
        \frac{\partial P}{\partial t_j} &= \frac{e_{k-1}(1^{n-\ell}, t_1,\dots,t_{j-1},t_{j+1},\dots,t_{\ell})}{\binom{n}{k}} \\
        &= \frac{k}{n}\cdot \frac{e_{k-1}(1^{n-\ell}, t_1,\dots,t_{j-1},t_{j+1},\dots,t_{\ell})}{\binom{n-1}{k-1}} = \frac{k}{n}\cdot \frac{e_{k-1}(1^{(n-1)-(\ell-1)}, t_1,\dots,t_{j-1},t_{j+1},\dots,t_{\ell})}{\binom{n-1}{k-1}}.
    \end{align*}
    Now focus on the term  $\frac{e_{k-1}(1^{(n-1)-(\ell-1)}, t_1,\dots,t_{j-1},t_{j+1},\dots,t_{\ell})}{\binom{n-1}{k-1}}$. We seek to invoke the inductive hypothesis on this term: note that $k'=k-1\ge 0$, $\ell'=\ell-1 \ge 0$. Since $\ell \ge 1$ and $\ell \le n/128$, we have $n \ge 128$, and hence $n-1 \ge 88$. Furthermore, it holds that
    \begin{align*}
        \ell - 1 \le \frac{n}{128}-1 \le \frac{n-1}{128}, \qquad k-1 \le \frac{n}{2}-1 \le \frac{n-1}{2}.
    \end{align*}
    So, by the inductive hypothesis, for any $t_1,\dots,t_{j-1},t_{j+1},\dots,t_\ell \in [-\theta, \theta]$, we have that the partial derivative with respect to $t_j$ is at least $\frac{k}{n}(0.12)^{\ell-1} > 0$. Summarily, we have that the partial derivative with respect to every $t_j$ is positive throughout the box $[-\theta, \theta]^\ell$. Thus, the minimum of $P$ on this box is achieved at $t_1=\dots=t_\ell=-\theta$. However, from our analysis above of this special case, we have that the value of this minimum is at least $(0.12)^\ell$. This completes the inductive step, and the proof of part (b) of the lemma.
\end{proof}

\subsection{Concluding by Bounding the Minimum Eigenvalue}
\label{sec:eigenvalue-bound}

In what follows, we will take $k=\lfloor n/2 \rfloor$ when instantiating \Cref{lemma:elementary-polynomial-bound}. The lemma implies that every $P_i$ from \eqref{eqn:P-i-in-terms-of-ell-i} satisfies that
\begin{align}
    \label{eqn:P-i-bounds}
    |P_i| \le (0.88)^{\ell_i}, \qquad \ell_i \le n/128 \implies P_i \ge (0.12)^{\ell_i} > 0.
\end{align}
Let us call a block $i \in [n]$ \textit{bad} if $\ell_i > n/128$, and \textit{good} otherwise. Let $R$ be the number of bad blocks. Recall the definition of $L$ from \eqref{eqn:def-ell-L}, which is the number of non-zero $\xi_u$ values. We will show that there is an absolute constant $\gamma > 0$ such that
\begin{align}
    \label{eqn:if-negative-L-big}
    \sum_{i=1}^n P_i < 0 \implies L \ge n^2/5000.
\end{align}
First, suppose that $R \ge n/2$. In this case, by the definition of a bad block, we immediately have
\begin{align}
    \label{eqn:L-lb-case-1}
    L = \sum_{i=1}^n \ell_i \ge \sum_{i\text{ bad}} \ell_i > Rn/128 \ge n^2/256.
\end{align}
Otherwise, consider $R < n/2$. Let $G=n-R > n/2$ be the number of good blocks, and define
\begin{align*}
    L_G = \sum_{i \text{ good}} \ell_i.
\end{align*}
Equation \eqref{eqn:P-i-bounds} gives a lower bound on $P_i$ whenever $i$ is good. Thus,
\begin{align}
    \label{eqn:good-block-lb}
    \sum_{i \text{ good}} P_i \ge \sum_{i \text{ good}} (0.12)^{\ell_i} \ge G\left(0.12\right)^{\frac{L_G}{G}},
\end{align}
where the last inequality is Jensen's inequality applied to the convex function $x \mapsto \left(0.12\right)^x$. On the other hand, for the bad blocks, we can use the upper bound from \Cref{lemma:elementary-polynomial-bound} part (a), as recorded in \eqref{eqn:P-i-bounds}, and get
\begin{align}
    \label{eqn:bad-block-ub}
    \sum_{i \text{ bad}} |P_i| \le \sum_{i \text{ bad}} (0.88)^{\ell_i} < R(0.88)^{n/128} < \frac{n}{2}(0.88)^{n/128},
\end{align}
where we used that $R < n/2$, and that $\ell_i > n/128$ for bad blocks.

Now, if $\sum_{i=1}^n P_i < 0$, then it must necessarily be the case that the contribution of bad blocks (in magnitude) is strictly larger than the contribution of the good blocks. That is, combining \eqref{eqn:good-block-lb} and \eqref{eqn:bad-block-ub}, it must hold that
\begin{align*}
    G\left(0.12\right)^{\frac{L_G}{G}} < \frac{n}{2}(0.88)^{n/128} \implies \left(0.12\right)^{\frac{L_G}{G}} < (0.88)^{n/128}. \tag{since $G > n/2$}
\end{align*}
Taking logs, we get
\begin{align*}
    \frac{L_G}{G} > \frac{n\log(1/0.88)}{128\log(1/0.12)} \ge  \frac{n}{2500} 
\end{align*}
Since $G > n/2$, this implies
\begin{align}
    \label{eqn:L-lb-case-2}
    L \ge L_G > n^2/5000.
\end{align}
Combining \eqref{eqn:L-lb-case-1} and \eqref{eqn:L-lb-case-2}, we get that
\begin{align*}
    \sum_{i=1}^n P_i < 0 \implies L \ge n^2/5000
\end{align*}
as desired.

We now finally translate the bound above to a bound on the minimum eigenvalue of the matrix $M$. Recall from Equation \eqref{eqn:normalized-eigenvalue} that every eigenvalue $\lambda_\xi$ of $M$ satisfies
\begin{align*}
     \frac{\lambda_\xi}{d} = \left(\prod_{u \in U} \frac{c_{\xi_u}}{c_0}\right)\frac{1}{n}\sum_{i=1}^n P_i.
\end{align*}
Note that $\prod_{u \in U} \frac{c_{\xi_u}}{c_0} > 0$. Now, suppose that $\lambda_\xi < 0$; then it must be the case that $\sum_{i=1}^{n}P_i < 0$. Equation \eqref{eqn:if-negative-L-big} then ensures that $L \ge n^2/5000$. In this case,
\begin{align*}
    \prod_{u \in U} \frac{c_{\xi_u}}{c_0} =  \left(\prod_{u\,:\,\xi_u = 0} \frac{c_{\xi_u}}{c_0}\right)\left(\prod_{u\,:\,\xi_u \neq 0} \frac{c_{\xi_u}}{c_0}\right) = \left(\prod_{u\,:\,\xi_u \neq 0} \frac{c_{\xi_u}}{c_0}\right) &\le \left(\max_{r \in \Z_Q \,:\, r \neq 0} \frac{c_r}{c_0}\right)^L \\
    &= \left(\frac{1+\frac{1}{2}\cos(2\pi/Q)}{3/2}\right)^L,
\end{align*}
where we used the expression for $c_r$ from \eqref{eqn:c_r}, and the fact that among $r \in \Z_Q, r \neq 0$, $\frac{1}{2}\cos(2\pi r/Q)$ is maximized at $r=1$ and $r=Q-1$. Defining $\rho:= \frac{1+\frac{1}{2}\cos(2\pi/Q)}{3/2} < 1$, we thus have that
\begin{align*}
    0 < \prod_{u \in U} \frac{c_{\xi_u}}{c_0} \le \rho^L.
\end{align*}
Recall also from \eqref{eqn:P-i-bounds} that every $|P_i| \le (0.88)^{\ell_i} \le 1$, and hence
\begin{align*}
    -1 \le \frac{1}{n}\sum_{i=1}^{n}P_i \le 1.
\end{align*}
We thus obtain that
\begin{align*}
    &\frac{\lambda_\xi}{d} = \left(\prod_{u \in U} \frac{c_{\xi_u}}{c_0}\right)\frac{1}{n}\sum_{i=1}^n P_i \ge - \prod_{u \in U} \frac{c_{\xi_u}}{c_0} \ge -\rho^L \ge -\rho^{n^2/5000} \\
    \implies \quad & \lambda_{\xi} \ge -d\rho^{n^2/5000}.
\end{align*}
If $\lambda_\xi \ge 0$, the bound above holds automatically. Thus, we have shown that every eigenvalue $\lambda_\xi$ corresponding to the eigenvectors $\chi_\xi$ is at least $-d\rho^{n^2/5000}$; since these eigenvectors form a basis, we have shown that the minimum eigenvalue $\lambda_{\min}$ of $M$ satisfies
\begin{align*}
    \lambda_{\min} \ge -d\rho^{n^2/5000}.
\end{align*}
But this is precisely what we wanted to show at the end of \Cref{sec:eig-bound-to-communication} to establish \Cref{thm:constant-gadget-lifting}: a lower-bound of $-d\rho^{\gamma n^2}$ on the smallest eigenvalue of $M$, where $0 < \rho < 1$ and $\gamma > 0$. Since the above analysis establishes this for a value of $\rho$ satisfying $0 < \rho < 1$ and $\gamma = 1/5000$, we conclude the proof of \Cref{thm:constant-gadget-lifting}.

\begin{remark}[General Constant-gadget Lifting]
    \label{remark:extend-lifting-to-all-functions}
    It is natural to wonder if the constant-gadget lifting theorem above could apply to all Boolean functions instead of just the unambiguous DNFs from \Cref{thm:unambigous-dnf}, similar to the all-purpose lifting theorem of \citet{goos2016rectangles}. The goal here would be to show that the constant-sized cyclic gadget lifts any function $f$ having $0$-certificate complexity $\C_0(f)$ to a communication problem $h$ satisfying $\log\Cov_0(h) \ge \Omega(\C_0(f))$. In our limited attempts, we were unable to extend the analysis above to derive such a general lifting theorem. In particular, the analysis above appears to crucially use the form of the DNF from \Cref{thm:unambigous-dnf} in defining the matrix $M$, and in expressing its eigenvalues in the form of elementary symmetric polynomials.
\end{remark}

\section{Quartic Separation between Certificate Complexity and Approximate Degree}
\label{sec:approx-degree-quartic-separation}
In this section, we restate and give a complete proof of \Cref{thm:approximate-degree}.
\ApproximateDegree*

\begin{proof}
    We will derive a family of functions $G:\{0,1\}^{O(n^3 \log n)} \to \{0,1\}$ that witnesses the desired separation from the unambiguous DNF $f$ given by \Cref{thm:unambigous-dnf}. Recall that we constructed a universe $U=B_1 \sqcup \dots B_n$, where each $|B_i|=n$ (so $|U|=n^2$), and a signed set-pair system $\mcT=\{(P_T, N_T)\}_T$, where $T=(i,S)$ ranges over all $i \in [n]$, and $S \subseteq B_i$, $|S|=k=\lfloor n/2 \rfloor$. We had that $\mcT$ is pairwise incompatible, and $|P_T|+|N_T|=O(n)$ for every $T$. The definition of $f$ was
    \begin{align*}
        f(x) = \bigvee_{(i,S)} \left(\bigwedge_{u \in P_{i,S}} x_u \wedge \bigwedge_{v \in N_{i,S}} \neg x_{v}\right).
    \end{align*}

    The input to the function $G$ will comprise of the variables $(x_u)_{u \in U}$, a pointer $p_{u}$ attached to every $x_{u}$, as well as an array of $n^2$ ``cells'' $K_1,\dots,K_{n^2}$. Every pointer points to one of the $n^2$ cells, and hence can be encoded using $O(\log n)$ bits. Each cell comprises of a \textit{certificate description}, which, as we shall see, will be an encoding of a sound and complete certificate for when $f(x)=1$. In particular, we will see how $f$ can be verified to evaluate to $1$ at the given $x$ by evaluating $O(n)$-many low-degree polynomials over the certificate description, and then taking the AND of these evaluations. Our high-level approach here is similar to that used in \cite[Corollary 6]{balodis2023unambiguous}. We note that it is essential that we use only $O(n)$ many low-degree evaluations, for the purposes of obtaining our final desired quartic separation.

    We first concretely specify the format of the certificate description. Allocating each integer $O(\log n)$ bits, a certificate description comprises of the following components:
    \begin{enumerate}
        \item[(1)] An integer $i$, integers $a_1,\dots,a_k$ (intended to constitute a positive term $P_{i,S} \subseteq B_i$), and integers $b_1,\dots,b_{n-k}$ (intended to be the ``within-block'' elements in the negative term $N_{i,S}$). This occupies $O(n \log n)$ bits in total.
        \item[(2)] A list of $n$ integers $t_1,\dots,t_n$ (where $t_j$ is intended to be an element of $S$, so that $\ell_j := \lambda^{\{i,j\}}_i(t_j)$ certifies that $m_{i,j}(S) \le \ell_j$). This occupies $O(n \log n)$ bits in total as well.
        \item[(3)] A list $\Gamma$ comprising of a total of $N=O(n)$ entries, where $N$ is one larger than the maximum term-width of the DNF in \Cref{thm:unambigous-dnf}. Each entry in the list is a pair of integers $(j, r)$ or a padding entry $\perp$ (encoded, e.g., as $\perp=(0,0)$). This list is intended to capture all the ``cross-block'' elements in the negative set $N_{i,S}$ listed down in lexicographic order, and occupies $O(n \log n)$ bits as well. %
    \end{enumerate}
    Thus, the total input size of $G$ is $n^2 + O(n^2 \log n) + n^2 \cdot O(n \log n) = O(n^3 \log n)$. We now specify how a verifier verifies a given certificate description. %

    For Item (1), the verifier first checks that the index $i$ belongs to $[n]$ (and is hence a bona fide block index). It then checks that all the integers $a_1,\dots,a_k$ and $b_1,\dots,b_{n-k}$ belong to $B_i$, and that $a_1 < \dots < a_k$ and $b_1 < \dots < b_{n-k}$ respectively. Thereafter, the verifier checks whether $x_{a_t}=1$ for every $t \in [k]$, and that $x_{b_t}=0$ for every $t \in [n-k]$ (since these variables are intended to be the variables in $P_{i,S}$ and $N_{i,S}$ respectively in the block $B_i$). Observe that the entire check here can be represented as an AND of $O(n)$-many deterministic decision tree evaluations each having depth $O(\log n)$ (since each evaluation reads a constant number of integers, and possibly an additional bit). Furthermore, if the check passes, the integers $a_1,\dots,a_k$ and $b_1,\dots,b_{n-k}$ are forced to partition $B_i$, since each list has distinct entries, their combined size is $n$, and no entry can occur in both lists since that would imply that the same input bit equals both $0$ and $1$. Consequently, we can identify the positive set $S = P_{i,S} := \{a_1,\dots,a_k\}$.

    For Item (2), the verifier checks, for every $t_j$ where $j \neq i$, that $t_j \in B_i$, and that $x_{t_j}=1$ ($t_i$ may be ignored). In this case, let $\ell_j := \lambda^{\{i,j\}}_i(t_j)$, and observe that $m_{i,j}(S) \le \ell_j$. Note that this check can also be performed as the AND of $O(n)$-many decision tree evaluations of depth $O(\log n)$.

    For Item (3), the verifier checks that the list $\Gamma$ is exactly
    \begin{align*}
        \underbrace{(j,1),(j,2),\dots,(j,\ell_j)}_{\text{for each $j \in [n] \setminus \{i\}$ concatenated in increasing order}}
    \end{align*}
    followed by at least one padding entry. If this is not the case, the verifier rejects. In addition, the verifier also checks that for every non-padding entry $(j,r)$ in the list, denoting $u=\left(\lambda^{\{i,j\}}_j\right)^{-1}(r)$, it holds that $x_u=0$. This check ensures that all the variables in $N_{i,S}$ outside the block $B_i$ are set to 0.

    It can also be seen (perhaps as a slightly tedious matter) that this entire operation is also an AND of $O(n)$-many, $O(\log n)$-depth decision tree evaluations performed on pairs of consecutive list entries. For example, let $J_{-i}=[n] \setminus \{i\}$. The verifier must first check that $\Gamma_1 =  (\min(J_{-i}), 1)$, and that $\Gamma_{N}=~\perp$. Now, for any non-padding entry $(j,r)$, the verifier can obtain $\ell_j := \lambda^{\{i,j\}}_i(t_j)$. Thereafter, defining the \textit{successor} function
    \begin{alignat*}{2}
        &\mathrm{succ}_i(j,r) &&= \begin{cases}
            (j, r+1), & \text{if $r < \ell_j$}, \\
            (j', 1), & \text{if $r=\ell_j$ and $j'$ is the next element in $J_{-i}$}, \\
            \perp, & \text{if $j=\max(J_{-i})$ and $r=\ell_j$},
        \end{cases} \\
        &\mathrm{succ}_i(\perp)&&=~\perp,
    \end{alignat*}
    it remains for the verifier to check that $\Gamma_{s+1}=\mathrm{succ}_i(\Gamma_s)$ for every $s \in [N-1]$. For each check here, note that the verifier reads only a constant number of integers (and possible additional bits).

    Combining all of the above, we obtain that the entire verification process is an AND of $O(n)$ many deterministic decision tree evaluations having depth $O(\log n)$.

    \paragraph{Soundness and Completeness.} Now suppose that the verifier accepts a certificate description. The structural checks from Item (1) ensure that $i \in [n]$, and that $S =\{a_1,\dots,a_k\} \subseteq B_i, |S|=k$. They also verify that all the variables in the positive term $P_{i,S}$ are set to 1, and all the variables in the negative term $N_{i,S}$ within the block $B_i$ are set to $0$. Next, if the checks from Item (2) pass, this confirms that $m_{i,j}(S) \le \ell_j$ for every $j \neq i$. Thereafter, the check in Item (3) ensures that there is a list entry $(j,r)$ for every $j\neq i$, $1 \le r \le \ell_j$. Since $\ell_j \ge m_{i,j}(S)$, the check on all the relevant variables being 0 ensure that for every $j \neq i$, the variables $x_u$ for $u$ in the set
    \begin{align*}
        A_{i,j}(S)=\{y \in B_j: \lambda^{\{i,j\}}_j(y) \le m_{i,j}(S)\}
    \end{align*}
    are set to 0. Together with the check from Item (1), this establishes that all the variables in the negative term $N_{i,S}$ are set to 0. This confirms that the DNF $f(x)$ evaluates to 1. 
    Thus, a certificate description accepted by the verifier soundly certifies $f(x)=1$. %

    Conversely, consider any input $x$ for which $f(x)=1$. Let $(i,S)$ be the term in the DNF that evaluates to 1, where $S=\{a_1 < \dots, a_k\}$ and $B_i \setminus S = \{b_1<\dots<b_{n-k}\}$ --- these can constitute Item (1) of a certificate description. For Item (2), we set $t_i$ arbitrarily, and for each $j \neq i$, we set $t_j$ to be the element $u \in S$ which attains $\lambda^{\{i,j\}}_i(u)=m_{i,j}(S)$. The construction of the DNF $f$ guarantees that $\sum_{j \neq i} m_{i,j}(S) < N$, and hence all the required $(j,r)$ entries fit into the list $\Gamma$ in Item (3), together with at least one padding entry as well; furthermore, the input bits corresponding to all these entries are 0 as well, since the term $(i,S)$ evaluates to 1. %
    Thus, all the checks of the verifier pass, and the verifier accepts correctly, establishing the completeness of the verifier.

    \paragraph{Defining the function $G$.} Recall that the input to the function $G$ comprises of $x:=(x_u)_{u \in U}$, pointers $p=(p_u)_{u \in U}$ associated with every $x_u$, and $n^2$ cells $K:= K_1,\dots,K_{n^2}$, each containing a certificate description. We now define what it means for a cell $K_z$ to be \textit{valid} for $(x,p)$. Consider the index $i$ and integers $a_1,\dots,a_k$ in Item (1) of $K_z$. Then, $K_z$ is valid for $(x,p)$ if:
    \begin{enumerate}
        \item[(a)] The certificate description in $K_z$ is accepted by the verifier, and
        \item[(b)] $p_{a_1}=z$. %
        In words, the pointer corresponding to the \textit{first} index $a_1$ in the positive set $S$ must point to the name of the cell.
    \end{enumerate}

    We note first that at most one cell $K_z$ can be valid for $(x,p)$. To see this, suppose that both $K_z$ and $K_{z'}$ are valid for $(x,p)$. This means that both $K_z$ and $K_{z'}$ are accepted by the verifier. By soundness and unambiguity of the DNF $f$, they both identify the same term $(i,S)$ that is satisfied in $f$ --- let $a_1$ be the smallest member of $S$. But now, the pointer check requires that $p_{a_1}=z$ and $p_{a_1}=z'$, necessitating $z=z'$.

    So now, for a given input $x,p,K$ to $G$, and any cell $z$, define $G_z = 1$ iff the cell $K_z$ is valid for $(x,p)$. The desired function $G$ is then defined as
    \begin{align*}
        G = \bigvee_{z=1}^{n^2} G_{z}.
    \end{align*}
    By our previous reasoning, the OR above is unambiguous: at most one $G_z$ can ever be 1. 
    
    Crucially also, each $G_z$ can be expressed as the AND of $O(n)$-many decision tree evaluations, each having depth $O(\log n)$. This holds true since evaluating $G_z$ requires (a) verifying the certificate $K_z$, which, as we already argued above, is an AND of $O(n)$ decision tree evaluations having depth $O(\log n)$), and (b) checking that $p_{a_1}=z$, which only requires reading $a_1$ and $p_{a_1}$, and is hence another decision tree evaluation of depth $O(\log n)$).

    This is useful because each $O(\log n)$-depth decision tree evaluation can be represented exactly by a polynomial of degree $O(\log n)$, and the AND function on $O(n)$ arguments can be $\eps$-approximated by a polynomial of degree $O(\sqrt{n\log(1/\eps)})$ \citep{nisan1994degree,de2008note}. Together, we get that the composition --- a polynomial of degree $O(\log n\sqrt{n\log(1/\eps)})$ --- $\eps$-approximates $G_z$.

    \paragraph{Lower-bounding $\C(G)$.} Consider the all-zeros input $\bzero$ to $G$, which satisfies that $G(\bzero)=0$. We will argue that $\bzero$ has large $0$-certificate complexity, namely
    \begin{align*}
        \C_0(G, \bzero) \ge n(n-k+1)=\Omega(n^2).
    \end{align*}
    Towards this, let $\rho$ be any partial assignment consistent with $\bzero$, and satisfying that $|\rho| < n(n-k+1)$. Call a position $u \in U$ \textit{touched} if $\rho$ fixes either $x_u$ or any bit associated with the pointer $p_u$. Then, since $|\rho| < n(n-k+1)$, there exists a block $i \in [n]$ that has at most $n-k$ touched positions. Choose any $k$ untouched positions in this block, and collect these in a set $S$. Also, since there are $n^2$ cells, there is at least one cell $z$ such that every entry in $\rho$ within the cell $z$ is equal to $*$. We now extend $\rho$ to an input $(x, p, K)$ as follows: we set $x_u=1$ for every $u \in S$, and $x_u=0$ for every $u \in U \setminus S$. Then, observe that $f(x)=1$. By the completeness of our certificate scheme, there exists a certificate description that certifies this; we fill the cell $K_z$ with this certificate description. Finally, we set the pointer $p_{\min(S)}=z$. By setting any remaining input bits in $(x,p,K)$ to zero, we obtain an input that is consistent with $\rho$. But by construction, the cell $K_z$ is valid for $(x,p)$, and hence $G(x,p,K)=1$. Thus, $\rho$ cannot originally be a $0$-certificate for $G$ at $\bzero$, giving that
    \begin{align}
        \label{eqn:certificate-cxty-lb-cheat-sheet-total-function}
        \C(G) \ge \C_0(G, \bzero) \ge \Omega(n^2).
    \end{align}

    \paragraph{Upper-bounding $\widetilde{\Deg}(G)$.}
    We will now show that $G=\vee_{z=1}^{n^2} G_z$ has small approximate degree. Recall, as mentioned above, that each $G_z$ can be $\eps$-approximated by a polynomial of degree at most $O(\log n\sqrt{n\log(1/\eps)})$. Choosing $\eps=1/(3n^2)$, we get that there is a polynomial $p_z$ of degree $O(\sqrt{n} \log^{3/2} n)$ such that
    \begin{align*}
        |p_z - G_z| \le 1/(3n^2)
    \end{align*}
    on all inputs. Then, define the polynomial
    \begin{align*}
        p = \sum_{z=1}^{n^2} p_z.
    \end{align*}
    The degree of $p$ is still $O(\sqrt{n} \log^{3/2} n)$. Now, observe that if $G=0$, then every $G_z=0$, and so
    \begin{align*}
        |p-G| = |p| \le \sum_{z=1}^{n^2} |p_z| = \sum_{z=1}^{n^2} |p_z-G_z| \le n^2 \cdot 1/(3n^2) = 1/3.
    \end{align*}
    On the other hand, if $G=1$, then exactly one $G_z=1$, and every other $G_z=0$. Thus,
    \begin{align*}
        |p-G| = |p-1| = \left|\sum_{z=1}^{n^2}(p_z - G_z)\right| \le  \sum_{z=1}^{n^2} |p_z-G_z|  \le n^2 \cdot 1/(3n^2) = 1/3.
    \end{align*}
    We therefore have that $p$ $1/3$-approximates $G$, meaning that
    \begin{align*}
        \widetilde{\Deg}_{1/3}(G) = \widetilde{\Deg}(G) \le O(\sqrt{n} \log^{3/2} n).
    \end{align*}

    \paragraph{Concluding the quartic separation.} We have thus constructed a total function $G$ for which 
    \begin{align*}
        \C(G) \ge \Omega(n^2), \qquad \widetilde{\Deg}(G) \le O(\sqrt{n} \log^{3/2} n).
    \end{align*}
    Thus, $G$ enjoys the desired quartic separation
    \begin{align*}
        \C(G) \ge \Omega\left(\frac{\widetilde{\Deg}(G)^4}{\log^6(n)}\right).
    \end{align*}

\end{proof}

\section{Lower Bound for Multiclass Sample Compression}
\label{sec:multiclass-sample-compression}

In this section, we restate and prove the lower bound on sample compression schemes for multiclass concept classes. For a formal definition of a sample compression scheme, we refer the reader to Definition 6 in \citet{pabbaraju2024multiclass}.

\Compression*
\begin{proof}
    We will show: for infinitely many $n$, there exists a multiclass concept class $\mcC$ that uses at most $2^{O(n^2)}$ labels and has Natarajan dimension 1, such that any sample compression scheme for $\mcC$ must have size at least $\Omega(n)$. The theorem statement then follows by reparameterizing.

    Our proof is largely based on the sample compression lower bound for multiclass concept classes shown in \cite{pabbaraju2024multiclass}; we refer the reader to that proof, while giving the necessary details below.

    Choose an $n$ for which \Cref{thm:alon-saks-seymour} holds: namely, there exists a graph $G=(V,E)$ having $2^{\Theta(n^2)}$ vertices that admits a biclique partition of size $2^{O(n)}$ but has chromatic number $2^{\Omega(n^2)}$. Fix a partition of $G$ into at most $2^{O(n )}$ bicliques. The domain $\mcX$ of the partial concept class $\mcP$ is the set of bicliques in this partition. Note that each biclique is determined by a set $L$ of left vertices and a set $R$ of right vertices, such that every vertex in $L$ is connected to every vertex in $R$. Now, for each vertex $v$ in $G$, we will have a partial concept $p_v$, which maps a biclique $Q=(L, R)$ in the partition as follows:
    \begin{align*}
        p_v(Q) = \begin{cases}
            0 & \text{if $v \in L$,} \\
            1 & \text{if $v \in R$,} \\
            * & \text{otherwise.}
        \end{cases}
    \end{align*}
    This partial class was shown by \cite{alon2022theory} to have VC dimension 1. Now consider the following total concept class $\mcC$: for every $p_v \in \mcP$, the class $\mcC$ has $c_v$ which maps any biclique $Q$ as
    \begin{align*}
        c_v(Q) = \begin{cases}
        p_v(Q) & \text{if $p_v(Q) \neq *$}, \\
        v & \text{otherwise.}
        \end{cases}
    \end{align*}
    By construction, $\mcC$ disambiguates $\mcP$ (i.e., for every $p_v \in \mcP$, there is a total concept $c_v \in \mcC$ consistent with $p_v$); furthermore, observe that it uses at most $|V|+2$ labels. \cite{pabbaraju2024multiclass} shows that the class $\mcC$ has \textit{DS dimension} 1, and since the Natarajan dimension is never larger than the DS dimension, its Natarajan dimension is also 1.

    Now let $(\kappa, \rho)$ be any sample compression scheme for $\mcC$ of size $k$. Since $\mcC$ disambiguates $\mcP$, $(\kappa, \rho)$ is also a sample compression scheme for $\mcP$. Notably, this means that the support of any partial concept $p_v$ (i.e., the domain points at which it is not $*$, but either 0 or 1) gets compressed by $\kappa$ to a subsequence of size at most $k$, and applying the reconstructor $\rho$ to this compressed subsequence recovers all the labels correctly on the support, meaning also that the reconstructor output disambiguates $p_v$. This means that, if we were to iterate over all possible subsequences of the domain, all possible labelings of the subsequence by a concept in $\mcP$, all possible side-information strings of size at most $k$, and collect the total concepts output by the reconstructor on all of these, then the resulting total concept class would disambiguate $\mcP$. The number of reconstructor outputs we need to consider is at most
    \begin{align*}
        |\mcX|^{O(k)} \cdot 2^{k} \cdot 2^{k} = 2^{O(kn)}.
    \end{align*}
    But \cite{alon2022theory} also show that any disambiguating class for $\mcP$ determines a valid coloring of the vertices in $G$ using a number of colors equal to the size of the disambiguating class. Since the chromatic number of $G$ is at least $2^{\Omega(n^2)}$, we thus have that
    \begin{align*}
        2^{O(kn)} \ge 2^{\Omega(n^2)} \implies k \ge \Omega(n).
    \end{align*}  
\end{proof}

\subsection*{Acknowledgements}
Various ideas used in the paper were individually developed over numerous lengthy rounds of interaction with ChatGPT and Codex. %
This work is supported by Moses Charikar's and Greg Valiant's Simons Investigator Awards, and a Google PhD Fellowship. I would like to thank Greg Valiant for encouraging me to think about multiclass sample compression again. %
I am also very grateful to Mika G{\"o}{\"o}s for suggesting a more direct construction for \Cref{thm:approximate-degree}, and for helpful comments that improved the exposition.

\bibliographystyle{plainnat} 
\bibliography{references}

\appendix
\section{Optimality of \Cref{corollary:partial-functions,corollary:intersecting-hypergraphs}}
\label{sec:optimality-partial-functions-intersecting-hypergraphs}

We give direct upper bounds showing that the separations in \Cref{corollary:partial-functions} and \Cref{corollary:intersecting-hypergraphs} are optimal up to constant factors.

\begin{claim}
    \label{claim:partial-functions-separation-optimality}
    For any partial function $f$ and for any $x \in f^{-1}(*)$, it holds that
    \begin{align*}
        \min\left\{\C_{\overline{0}}(f, x), \C_{\overline{1}}(f, x)\right\} \le \C(f)^2.
    \end{align*}
\end{claim}
\begin{proof}
    Fix $f$, and any $x \in f^{-1}(*)$. Let $k := \C(f)$. For every $0$-input to $f$, choose a $0$-certificate of size at most $k$, and let $\mcC_0$ be the resulting collection of all the $0$-certificates. Similarly, for every $1$-input to $f$, choose a $1$-certificate of size at most $k$, and let $\mcC_1$ be the resulting collection of all the $1$-certificates.

    For any certificate $\rho$, define its disagreement set with $x$ as
    \begin{align*}
        D_x(\rho) := \{i \in \dom(\rho) : \rho_i \neq x_i\}.
    \end{align*}
    Since $f(x)=*$, it must be the case that $D_x(\rho) \neq \emptyset$ for every $\rho \in \mcC_0 \cup \mcC_1$.

    Now, let $\tau$ be a certificate that is consistent with $x$, for which $\dom(\tau)$ intersects every $D_x(\rho)$ for $\rho \in \mcC_0$. Then, observe that $\tau$ is a $\overline{0}$-certificate for $x$. Otherwise, there would be an extension of $\tau$, that is consistent with $\tau$, and is a 0-input to $f$, which means it should be consistent with its $0$-certificate in $\mcC_0$. But since $\tau$ is inconsistent with every $\rho \in \mcC_0$, this is a contradiction. By a similar argument, we have that any $\beta$ that is consistent with $x$, and for which $\dom(\beta)$ intersects every $D_x(\sigma)$ for $\sigma \in \mcC_1$ is a $\overline{1}$-certificate for $x$.

    So, consider the family of sets
    \begin{align*}
        \mcD_0 := \{\mcD_x(\rho):\rho \in \mcC_0\}, \qquad \mcD_1 := \{\mcD_x(\sigma):\sigma \in \mcC_1\}.
    \end{align*}
    Suppose that the size of the largest subfamily of $\mcD_0$ of pairwise disjoint sets, i.e., the \textit{matching number} of $\mcD_0$, is at most $k$, and consider any maximal subfamily comprising of $\mcD_x(\rho_1),\dots, \mcD_x(\rho_t)$ of pairwise disjoint sets, where $t \le k$. This means that the union
    \begin{align*}
        \mcD_x(\rho_1) \cup \dots \cup \mcD_x(\rho_t)
    \end{align*}
    hits every member of $\mcD_0$; its size is at most $tk \le k^2$. Thus, in this case,
    \begin{align*}
        \C_{\overline{0}}(f, x) \le k^2.
    \end{align*}
    In the other case, there exist $k+1$ sets $\mcD_x(\rho_1),\dots, \mcD_x(\rho_{k+1})$ that are pairwise disjoint, where every $\rho_i \in \mcC_0$.

    Let $T$ be defined as
    \begin{align*}
        T  := \bigcup_{j=1}^{k+1} \dom(\rho_j) \setminus D_x(\rho_j).
    \end{align*}
    Since $\dom(\rho_j)$ has size at most $k$, and every $D_x(\rho_j)$ is non-empty, we have that
    \begin{align*}
        \left|\dom(\rho_j) \setminus D_x(\rho_j)\right| \le k-1.
    \end{align*}
    Thus,
    \begin{align*}
        |T| \le (k+1)(k-1) = k^2-1.
    \end{align*}
    We now claim that $T$ hits every member of $\mcD_1$. Suppose not, meaning that for some $\sigma \in \mcC_1$,
    \begin{align*}
        T \cap D_x(\sigma) = \emptyset.
    \end{align*}
    Now, each $\rho_j$ is a $0$-certificate for $f$, whereas $\sigma$ is a $1$-certificate for $f$; thus, every $\rho_j$ must be inconsistent with $\sigma$. However, $\rho_j$ and $\sigma$ cannot disagree at a coordinate in $\dom(\rho_j) \setminus D_x(\rho_j)$, since these are the coordinates at which $\rho_j$ agrees with $x$, and hence, if $\sigma$ disagrees with a coordinate within this set, it would disagree with $x$ at that coordinate, contradicting that $T \cap D_x(\sigma)=\emptyset$. Thus, every $\rho_j$ has a disagreement with $\sigma$ within $D_x(\rho_j)$.

    But the sets $D_x(\rho_1),\dots,D_x(\rho_{k+1})$ are pairwise disjoint. So, $\sigma$ having a disagreement within each of these sets would in particular imply that $|\sigma| \ge k+1$, which contradicts $|\sigma| \le k$. We conclude that $T$ hits every member of $\mcD_1$, which implies 
    \begin{align*}
        \C_{\overline{1}}(f, x) \le k^2-1.
    \end{align*}
    Combining the two cases above, we obtain
    \begin{align*}
        \min\left\{\C_{\overline{0}}(f, x), \C_{\overline{1}}(f, x)\right\} \le k^2,
    \end{align*}
    and the claim follows by recalling that $k=\C(f)$.
\end{proof}

\begin{claim}
    \label{claim:intersecting-hypergraphs-separation-optimality}
    For any intersecting hypergraph $G=(V, E)$ and for any coloring $c:V \to \{0,1\}$, there exists a $c$-monochromatic hitting set that has size at most $r(G)^2$.
\end{claim}
\begin{proof}
    Let $G=(V,E)$ be any intersecting hypergraph having rank $r(G)$, and let $c:V \to \{0,1\}$ be any coloring of its vertices. We will show that there is either a $0$-monochromatic or $1$-monochromatic hitting set of size at most $r(G)^2$.

    First, observe that if any edge is monochromatic, we are done, since the hypergraph is intersecting, which implies that the edge itself is a monochromatic hitting set of size at most $r(G)$. So, let us assume that every edge contains at least one $0$-vertex and one $1$-vertex.

    For each edge $e$, let $A(e)$ comprise the $0$-vertices, and $B(e)$ comprise the $1$-vertices in $e$. Consider the family of sets
    \begin{align*}
        \mcA := \{A(e): e \in E\}.
    \end{align*}
    Note that every set in $\mcA$ is non-empty, and has size at most $r(G)-1$. Let $\nu(\mcA)$ be the matching number of $\mcA$, i.e., the largest number of pairwise disjoint members in $\mcA$. There are two cases: 

    \paragraph{Case 1: $\nu(\mcA) \le r(G)$.}
    Let $A(e_1),\dots, A(e_t)$ be a maximal subfamily of $\mcA$ that is pairwise disjoint, where $t \le r(G)$. By maximality, it holds that the union
    \begin{align*}
        U_0 = A(e_1) \cup \dots \cup A(e_t) 
    \end{align*}
    hits every member of $\mcA$. We thus have that $U_0$ is a $0$-monochromatic set of size at most $t\cdot r(G) \le r(G)^2$ that hits every edge in the hypergraph.
    
    \paragraph{Case 2: $\nu(\mcA) > r(G)$.} Let $A(e_1),\dots, A(e_{r(G)+1})$ be pairwise disjoint. Consider the set
    \begin{align*}
        U_1 = B(e_1) \cup \dots \cup B(e_{r(G)+1}).
    \end{align*}
    Again, every $|B(e_i)| \le r(G)-1$, and hence $|U_1| \le (r(G)+1)(r(G)-1) = r(G)^2-1$.
    We now claim that $U_1$ hits every edge in $E$. Suppose not, meaning that for some edge $f$, $U_1 \cap f = \emptyset$. Since $G$ is intersecting, $f$ has non-empty intersection with every $e_i$. But since $f \cap U_1 = \emptyset$, it must be the case that for every $i$,
    \begin{align*}
        \emptyset \neq f \cap e_i \subseteq A(e_i).
    \end{align*}
    Since $A(e_1),\dots,A(e_{r(G)+1})$ are pairwise disjoint, we obtain that $f$ contains at least $r(G)+1$ vertices, contradicting that $|f| \le r(G)$. Thus, $U_1$ is a $1$-monochromatic set that hits every edge in the hypergraph.

    Combining the two cases, we have shown that $G$ always contains a $c$-monochromatic hitting set of size at most $r(G)^2$.
\end{proof}

\section{Proof of \Cref{thm:alon-saks-seymour}}
\label{sec:alon-saks-seymour-proof}

We restate and give the necessary proof details of \Cref{thm:alon-saks-seymour}.
\AlonSaksSeymour*
\begin{proof}
    We first use our lifting theorem to lift the unambiguous DNF separation of \Cref{thm:unambigous-dnf} to a separation in communication complexity. Namely, \Cref{thm:constant-gadget-lifting} gives a gadget $g: \{0,1\}^{3} \times \{0,1\}^3 \to \{0,1\}$ such that the communication problem $f \circ g^{n^2}: \{0,1\}^{3n^2} \times \{0,1\}^{3n^2} \to \{0,1\}$ obtained by lifting the DNF from \Cref{thm:unambigous-dnf} with the gadget $g$ satisfies
    \begin{align*}
        \log \Cov_0(f \circ g^{n^2}) \ge \Omega(n^2).
    \end{align*}
    Furthermore, it holds that
    \begin{align*}
        \log \Par_1(f \circ g^{n^2}) \le O(\UC_1(f)),
    \end{align*}
    since for any $n$-variate function $f$ and any $k$-sized gadget, it holds that $\log \Par_1(f \circ g^{n}) \le O(k \cdot \UC_1(f))$. Combining the two bounds, we obtain a function $h:\{0,1\}^{\Theta(n^2)} \times \{0,1\}^{\Theta(n^2)} \to \{0,1\}$ for which
    \begin{align*}
        \log \Cov_0(h) \ge \Omega(n^2), \qquad \log \Par_1(h) \le O(n).
    \end{align*}
    We can now directly apply Lemma 3.5 from \cite{cheung2023online} to the function $h$, which converts $h$ into a graph $G$ having $2^{\Theta(n^2)}$ vertices, and satisfying 
    \begin{align*}
        \bp(G) \le \Par_1(h)^2 \le 2^{O(n)}, \qquad \chi(G) \ge \sqrt{\Cov_0(h)} \ge 2^{\Omega(n^2)}.
    \end{align*}
\end{proof}

\section{Proof of \Cref{corollary:sensitivity}}
\label{sec:sensitivity-proof}

We restate and give the necessary proof details for \Cref{corollary:sensitivity}.

\Sensitivity*
\begin{proof}
    The proof is essentially contained in the proofs of Lemma 12 and Theorem 1 in \cite{ben2017low}; we flesh out the main steps here. Let $f$ be the unambiguous DNF from \Cref{thm:unambigous-dnf}, for which
    \begin{align*}
        \C_0(f) \ge \Omega(n^2), \qquad \UC_1(f) \le O(n).
    \end{align*}
    In fact, since $\C_0(f) \le \UC_1(f)^2$ always, we have that $\UC_1(f)=\Theta(n)$. Let $u := \UC_1(f)$; we have that $u=\Theta(n)$.

    Now, we apply the \textit{desensitizing transformation} introduced in Definition 11 in \cite{ben2017low} to $f$. Namely, consider the function $f':\{0,1\}^{3n^2} \to \{0,1\}$, which takes three inputs $x,y,z \in \{0,1\}^{n^2}$ and outputs
    \begin{align*}
        f'(x,y,z)=1 \text{ iff } &f(x)=f(y)=f(z)=1 \text{ and the unique term that } \\ &\text{evaluates to 1 is the same for $x,y$ and $z$.}
    \end{align*}
    Lemma 12 in \cite{ben2017low} shows that
    \begin{align*}
        \C_0(f') \ge \C_0(f).
    \end{align*}
    Now define the function $F := \mathrm{OR}_{3u} \circ f'$, i.e., $F$ takes $3u$ inputs to $f'$ and outputs their OR. \cite{ben2017low} show (in the proof of their Theorem 1) that 
    \begin{align}
        \label{eqn:sensitivity-ub}
        \s(F) \le 3u.
    \end{align}
    Now consider the 0-certificate complexity of $F$. To certify that $F=0$, one must certify that each of the $3u$ evaluations of $f'$ is $0$. Thus,
    \begin{align}
        \label{eqn:0-complexity-lb}
        \C_0(F) = 3u \cdot \C_0(f') \ge 3u \cdot \C_0(f) = \Omega(u^3).
    \end{align}
    Combining \eqref{eqn:sensitivity-ub} and \eqref{eqn:0-complexity-lb}, we get that $\C_0(F) \ge \Omega(\s(F)^3)$. This completes the proof.
\end{proof}

\end{document}